\documentclass[11pt]{article}
\usepackage{amsmath,amssymb,amsthm,amsfonts,bm}
\usepackage{geometry}
\usepackage{hyperref}

\newtheorem{theorem}{Theorem}
\newtheorem{remark}{Remark}
\newtheorem{proposition}{Proposition}
\newtheorem{lemma}{Lemma}
\newtheorem{definition}{Definition}
\newtheorem{corollary}{Corollary}
\newtheorem{example}{Example}

\begin{document}

\title{Double Covariance Model for Entangled Quantum States: Gaussian Reduction to Second Order Covariances}
\author{Andrei Khrennikov\\
Center for Mathematical Modeling in Physics and Cognitive Sciences\\ Institute for Mathematics and Physics,
Linnaeus University, V\"axj\"o, SE-351 95, Sweden}
\date{}
\maketitle

\abstract{We propose a novel approach to the problem of interconnecting the probabilistic formalisms of classical and quantum physics, focusing on its most challenging aspect: the classical probabilistic generation of entangled states. We show that the statistics encoded in the density operators of composite quantum systems correspond to fourth-order classical statistics. Specifically, to generate a density operator, one must consider the covariance of a random covariance operator. We term this framework the Double Covariance Model (DCM). This double covariance possesses a non-trivial internal structure that arises from the interplay between two distinct time scales, combining temporal and statistical covariances.
In this article, we exploit a well-known property of Gaussian processes: the second-order moment determines the moments of higher orders, specifically the fourth-order moment. This Gaussian reduction simplifies the DCM by reducing it to second-order statistics. Utilizing (circular) Gaussian processes simplifies and generalizes the DCM construction for entangled states, rendering it mathematically rigorous. Furthermore, it clarifies the classical probabilistic meaning of concurrence, a foundational quantitative measure of entanglement.}

{\bf keywords:} classical vs quantum probability,  double covariance model, quantum entangled state, classical Gaussian processes, fourth-order vs. second order statistics, micro- and macro-time scales 

\section{Introduction}

Understanding the interrelation between the probabilistic formalisms of classical and quantum physics remains a complex foundational problem. This topic has been analyzed and explored from various perspectives\footnote{For foundational overviews, see the monographs \cite{khrennikov2009,khrennikov2016}. Key mathematical works reviewing the structural interplay between classical and quantum probability include \cite{Accardi2000, AccardiFidaleo2002, gudder1988, Gudder1998Anomalies, meyer1993, Obata2004}; for epistemological analyses, see \cite{Plotnitsky2010, Plotnitsky2022,Kup,Red1,Surace}. For physical frameworks exploring classical probabilistic structures beyond standard quantum theory, see, for example,  \cite{Mad27,DeB60,Boh52,Dur09,AM1,AM3,LA2,hooft1,hooft2,Mhooft,Elze1, Elze2, Bush1,Bush2a,Bush2b,Wya05,ALA,PCSFT1}.}  One of the most challenging issues is the possibility of classical probabilistic modeling of the states of composite quantum systems, with a specific emphasis on the generation of entangled states. A new approach to this problem was proposed in \cite{DCM} within the framework of the {\it Double Covariance Model} (DCM), wherein density operators are classically generated as covariances of covariances. 

In practice, this classical-to-quantum transition is quite subtle. The DCM is based on the interplay between two distinct time scales:
\begin{itemize}
\item a fine time scale (subquantum), 
\item a rough time scale (quantum).    
\end{itemize}
The quantum scale corresponds to the scale of physical measurements, whereas the subquantum time scale represents the ontic time scale—the scale of reality as it is; see \cite{Markov} for the coupling with the hydrodynamic limit model; cf. \cite{ALA}). 

The first-order covariance operator $\hat C_\Delta$ is obtained via temporal correlation over the interval $\Delta$, which determines the measurement time scale. This temporal covariance is a random operator $\hat C_\Delta = \hat C_\Delta(\omega)$, where the randomness represents statistical randomness described by the Kolmogorov probability model. Ultimately, the DCM operates with the statistical covariance operator of this random operator, thereby employing fourth-order statistics. This construction leads to the ``double covariance operator'' $\widehat C: H_A \otimes H_B \to H_A \otimes H_B.$ The classical-to-quantum transition from the double covariance operator $\widehat C$ to the corresponding density operator $\hat \rho_{C}$ is achieved via trace normalization:
\begin{equation}
\label{TR}
\widehat C \to \hat \rho_{C} = \widehat C / \rm{Tr} \widehat C.
\end{equation}    
All density operators, including those of entangled states, can be generated in this manner. Further development of the DCM 
\cite{Markov} led to the derivation of quantum Markovian dynamics by exploring the framework of the hydrodynamic limit.

In light of the above considerations, the DCM shapes entangled quantum states on the basis of fourth-order classical statistics—specifically, the covariance of covariances. In this paper, we exploit a well-known property of Gaussian processes: the second-order covariance completely determines the covariances of higher orders, including the fourth-order covariance used in the DCM. Our primary aim is the Gaussian reduction of the DCM to second-order covariances, which enables the generation of entangled states using classical Gaussian processes. Utilizing these processes simplifies and generalizes the foundational construction proposed in 
\cite{DCM}, rendering it mathematically rigorous. Furthermore, this Gaussian formalism provides a subquantum interpretation of concurrence—one of the primary quantitative measures of entanglement—in terms of energy redistribution.

The classical (subquantum) states of the subsystems $S_A$ and $S_B$ of a composite quantum system $S_{AB}$ are mathematically described by circular Gaussian stochastic processes $X_t$ and $Y_t$ taking values in their corresponding Hilbert spaces. Surprisingly, the Gaussian DCM analysis demonstrates that entangled quantum states can be generated by statistically independent Gaussian processes. Within the DCM, the fundamental origin of entanglement is not statistical dependence, but rather temporal synchronization at the subquantum time scale. This synchronization can be realized via various orthogonal decompositions of the space $L_2(\Delta)$; such decompositions determine the synchronized energy redistribution of the subquantum signals $X_t$ and $Y_t$.

We emphasize that DCM is in its early stages of development and does not aim to provide a comprehensive classical probabilistic framework for all of quantum mechanics. At present, DCM addresses a specific foundational challenge: constructing a viable classical probabilistic model that accounts for quantum entanglement. Beyond this foundational role, DCM offers immediate pragmatic utility for the classical probabilistic generation of entangled states, with applications ranging from quantum-inspired computing  to quantum-like modeling in cognition and decision-making 
\cite{Bio}. Finally, we highlight a key foundational characteristic of DCM: broad classes of distinct classical states (stochastic processes) can map to the exact same quantum state, a property that shares conceptual ground with the work of 't Hooft \cite{hooft1,hooft2} (see also 
\cite{Mhooft}). 

\section{Mathematical and Physical Preliminaries}

\subsection{Physics}

The Double Covariance Model (DCM) is intended as a classical probabilistic
framework for the generation of quantum states of composite systems.
We consider a bipartite quantum system
\[
S_{AB}=S_A+S_B,
\]
with subsystem state spaces represented by complex Hilbert spaces
\[
H_A,\qquad H_B.
\]
The corresponding quantum state space of the composite system is the tensor
product
\[
H_A\otimes H_B.
\]
In the DCM, the quantum state of the composite system is not taken as
fundamental. Instead, it is generated from a deeper classical probabilistic
description at a finer time scale. The basic subquantum variables are
stochastic processes
\[
X=\{X_t\}_{t\in\Delta},\qquad
Y=\{Y_t\}_{t\in\Delta},
\]
taking values in the subsystem Hilbert spaces \(H_A\) and \(H_B\), respectively.
These processes represent the classical states of the subsystems \(S_A\) and
\(S_B\) at the subquantum time scale.

The DCM is based on the interplay of two distinct levels of description.

\begin{itemize}
\item[(i)] At the \emph{subquantum} time scale, the processes \(X_t\) and \(Y_t\)
describe the fine temporal dynamics of the subsystems.

\item[(ii)] At the \emph{quantum} or observational time scale, one does not
observe the instantaneous values of these processes. Instead, one forms a
temporal covariance over a time window \(\Delta\), and then takes the
statistical covariance of this random temporal covariance.
\end{itemize}

Thus the DCM is built from two successive averaging procedures:
first a temporal averaging over the observation window \(\Delta\), and then a
statistical averaging over the probability space of random parameters.
This is why the resulting object is a \emph{double covariance}: it is a
fourth-order classical statistical object constructed as the covariance of a
random covariance operator.

The physical significance of the model is that it provides a classical
probabilistic mechanism for the generation of quantum states, including
entangled states. In particular, the DCM suggests that quantum entanglement
need not originate from ordinary second-order statistical dependence between
the subquantum processes \(X\) and \(Y\). Instead, it may arise from a more
subtle temporal organization of the subquantum signals within the observation
window \(\Delta\). In the Gaussian setting studied in this paper, this temporal
organization is encoded in the geometry of the Karhunen--Lo\`eve modes and in
the integrated tensors generated by them.

As was emphasized, the DCM is based on the interplay between two time scales: a subquantum
(microscopic) time scale and a quantum (macroscopic) time scale. We denote the
corresponding time variables by \(t\) and \(\tau\), respectively. The length
\(|\Delta|\) of the observation window plays the role of a unit of macroscopic
time. In the present paper we work only with the subquantum time variable \(t\).
We do not study the evolution of the quantum state
\(\hat\rho=\hat\rho(\tau)\); for the derivation of quantum Markov dynamics
within the DCM by means of the hydrodynamic limit formalism, see
\cite{dynamics}. Our focus here is on the generation of density operators at a
fixed macroscopic time scale.

The quantum time scale is interpreted as the time scale of observational
dynamics. In this respect, we follow Bohr's view of quantum mechanics as a
theory of measurement outcomes and their probabilistic structure, rather than a
direct description of underlying microscopic processes \cite{plotnitsky}.

The DCM is related only indirectly to the traditional foundational programme of
hidden-variable theories \cite{Bell1,Bell2}. The subquantum stochastic processes that
represent classical states in the DCM should not be identified with hidden
variables in the conventional sense. Rather, the DCM can be viewed as a natural
development of Prequantum Classical Statistical Field Theory (PCSFT)
\cite{PCSFT1}, a theory of classical random fields underlying quantum phenomena.

At the same time, the DCM differs essentially from PCSFT. PCSFT is formulated in
terms of second-order moments of random fields, whereas the DCM is based on fourth-order moments - a
double-covariance construction and, crucially, on the interplay between
microscopic and macroscopic time scales.

\subsection{Mathematics}

All Hilbert spaces in this paper are complex. For Hilbert spaces \(H_1\) and
\(H_2\), we denote by \(L(H_1,H_2)\) the space of continuous linear operators
from \(H_1\) to \(H_2\). In particular,
$
L(H):=L(H,H).
$
(In the finite dimensional case these are simply spaces of linear operators.)

Throughout the paper, \(H_A\) and \(H_B\) denote finite-dimensional complex
Hilbert spaces associated with the two subsystems. We shall repeatedly use the
canonical identification of the tensor product space \(H_A\otimes H_B\) with
the space of linear operators \(L(H_B,H_A)\), given on simple tensors by
\[
a\otimes b \longleftrightarrow |a\rangle\langle b|,
\qquad a\in H_A,\; b\in H_B.
\]
Accordingly, we will freely pass between vectors in \(H_A\otimes H_B\) and the
corresponding operators in \(L(H_B,H_A)\).

Let \((\Omega,\mathcal F,P)\) be a Kolmogorov probability space. Here $\Omega$ is a set of random parameters (``elementary events''), 
${\cal F}$ is a $\sigma$-algebra of its subsets representing events, and $P$ is a probability measure on $\Omega.$ 
Let
$
H_A,\;  H_B
$
be (finite-dimensional) complex Hilbert spaces associated with the subsystems
\(S_A\) and \(S_B\), respectively.
We consider \(H_A\)- and \(H_B\)-valued stochastic processes
$
X=\{X_t\}_{t\in\Delta},\;
Y=\{Y_t\}_{t\in\Delta},
$
defined on \((\Omega,\mathcal F,P)\), where \(\Delta\subset\mathbb R\) is a
finite observation interval.

\subsubsection*{Centered circular Gaussian processes}

We shall assume throughout that the processes under consideration are
\emph{centered circular Gaussian}.

\begin{definition}
Let \(H\) be a complex separable Hilbert space and let
\[
X=\{X_t\}_{t\in\Delta}
\]
be an \(H\)-valued stochastic process.
We say that \(X\) is a \emph{centered circular Gaussian process} if for every
\(n\ge1\), every choice of times
\[
t_1,\dots,t_n\in\Delta,
\]
and every choice of vectors
\[
h_1,\dots,h_n\in H,
\]
the complex random vector
\[
\bigl(
\langle h_1,X_{t_1}\rangle,\dots,\langle h_n,X_{t_n}\rangle
\bigr)\in\mathbb C^n
\]
is jointly complex Gaussian, centered,
\[
\mathbb E\langle h_j,X_{t_j}\rangle=0,
\qquad j=1,\dots,n,
\]
and circularly symmetric in the sense that
\[
e^{i\theta}
\bigl(
\langle h_1,X_{t_1}\rangle,\dots,\langle h_n,X_{t_n}\rangle
\bigr)
\stackrel{d}{=}
\bigl(
\langle h_1,X_{t_1}\rangle,\dots,\langle h_n,X_{t_n}\rangle
\bigr)
\]
for every \(\theta\in [0, 2\pi)\).
\end{definition}

For centered complex Gaussian processes, circularity is equivalent to the
vanishing of all pseudocovariances:
\[
\mathbb E\bigl[\langle h,X_t\rangle\,\langle g,X_s\rangle\bigr]=0,
\qquad
h,g\in H,\quad t,s\in\Delta.
\]
Hence the law of a centered circular Gaussian process is completely determined
by its covariance kernel.

\subsubsection*{Covariance kernels of the pair \((X,Y)\)}

For the pair of processes \(X\) and \(Y\), we define the operator-valued
second-order covariance kernels
\[
\hat R_{XX}(t,s)=\mathbb E\bigl[\,|X_t\rangle\langle X_s|\,\bigr]
\in L(H_A),
\]
\[
\hat R_{YY}(t,s)=\mathbb E\bigl[\,|Y_t\rangle\langle Y_s|\,\bigr]
\in L(H_B),
\]
\[
\hat R_{XY}(t,s)=\mathbb E\bigl[\,|X_t\rangle\langle Y_s|\,\bigr]
\in L(H_B,H_A),
\]
\[
\hat R_{YX}(t,s)=\mathbb E\bigl[\,|Y_t\rangle\langle X_s|\,\bigr]
\in L(H_A,H_B).
\]

These kernels form the block covariance kernel of the joint process
\[
(X_t,Y_t):
\Omega\to H_A\oplus H_B,
\]
namely
\[
\hat \Gamma(t,s)=
\begin{pmatrix}
\hat R_{XX}(t,s) & \hat R_{XY}(t,s)\\[1mm]
\hat R_{YX}(t,s) & \hat R_{YY}(t,s)
\end{pmatrix}.
\]

If \(X\) and \(Y\) are jointly Gaussian and
\[
\hat R_{XY}(t,s)\equiv 0
\qquad (t,s\in\Delta),
\]
then the processes \(X\) and \(Y\) are independent. This fact will be used
repeatedly below.

\subsubsection*{Temporal Hilbert spaces and covariance operators}

To formulate the Karhunen--Lo\`eve decomposition, we introduce the temporal
Hilbert spaces
\[
\mathcal H_A:=L^2(\Delta;H_A),
\qquad
\mathcal H_B:=L^2(\Delta;H_B),
\]
with inner products
\[
\langle u,v\rangle_{\mathcal H_A}
=
\int_\Delta \langle u(t),v(t)\rangle_{H_A}\,dt,
\]
\[
\langle u,v\rangle_{\mathcal H_B}
=
\int_\Delta \langle u(t),v(t)\rangle_{H_B}\,dt.
\]

A process \(X=\{X_t\}_{t\in\Delta}\) with
\[
\mathbb E\int_\Delta \|X_t\|_{H_A}^2\,dt<\infty
\]
may be viewed as a random element of \(\mathcal H_A\). Its covariance operator
is the positive trace-class operator
\[
\hat T_X:\mathcal H_A\to\mathcal H_A
\]
defined by
\[
\langle \hat T_Xu,v\rangle_{\mathcal H_A}
=
\mathbb E\bigl[
\langle X,u\rangle_{\mathcal H_A}\,
\overline{\langle X,v\rangle_{\mathcal H_A}}
\bigr],
\qquad
u,v\in\mathcal H_A.
\]
Equivalently, \(\hat T_X\) is the integral operator with kernel \(\hat R_{XX}\):
\[
(\hat T_Xu)(t)=\int_\Delta \hat R_{XX}(t,s)u(s)\,ds.
\]

Similarly, the covariance operator of \(Y\) is the positive trace-class
operator
\[
\hat T_Y:\mathcal H_B\to\mathcal H_B
\]
defined by
\[
\langle \hat T_Yu,v\rangle_{\mathcal H_B}
=
\mathbb E\bigl[
\langle Y,u\rangle_{\mathcal H_B}\,
\overline{\langle Y,v\rangle_{\mathcal H_B}}
\bigr],
\qquad
u,v\in\mathcal H_B,
\]
or equivalently,
\[
(\hat T_Yu)(t)=\int_\Delta \hat R_{YY}(t,s)u(s)\,ds.
\]

\subsubsection*{Karhunen--Lo\`eve modes}

Since \(\hat T_X\) and \(\hat T_Y\) are positive trace-class operators, they admit
spectral decompositions
\[
\hat T_X=\sum_{k\ge1}\lambda_k\,|u_k\rangle\langle u_k|,
\qquad
\lambda_k\ge0,
\]
\[
\hat T_Y=\sum_{\ell\ge1}\mu_\ell\,|v_\ell\rangle\langle v_\ell|,
\qquad
\mu_\ell\ge0,
\]
where
\[
\{u_k\}\subset\mathcal H_A,
\qquad
\{v_\ell\}\subset\mathcal H_B
\]
are orthonormal families. These functions are the
\emph{Karhunen--Lo\`eve modes} of the processes \(X\) and \(Y\).

Because \(X\) and \(Y\) are centered circular Gaussian processes, their laws are
completely determined by \(\hat T_X\) and \(\hat T_Y\), and the corresponding
Karhunen--Lo\`eve expansions take the form
\[
X_t(\omega)=\sum_{k\ge1}\sqrt{\lambda_k}\,\xi_k(\omega)\,u_k(t),
\]
\[
Y_t(\omega)=\sum_{\ell\ge1}\sqrt{\mu_\ell}\,\eta_\ell(\omega)\,v_\ell(t),
\]
where \(\{\xi_k\}\) and \(\{\eta_\ell\}\) are standard circular complex Gaussian
random variables satisfying
\[
\mathbb E[\xi_k]=0,
\qquad
\mathbb E[\xi_k\overline{\xi_m}]=\delta_{km},
\qquad
\mathbb E[\xi_k\xi_m]=0,
\]
\[
\mathbb E[\eta_\ell]=0,
\qquad
\mathbb E[\eta_\ell\overline{\eta_m}]=\delta_{\ell m},
\qquad
\mathbb E[\eta_\ell\eta_m]=0.
\]
If \(X\) and \(Y\) are independent, then the two families
\(\{\xi_k\}\) and \(\{\eta_\ell\}\) are independent as well.

The KL modes will play a central role in the Gaussian reduction of the DCM and
in the constructive realization of entangled states.

\subsubsection*{The double covariance construction}

We now recall the basic DCM construction. Define the tensor-valued process
\[
Z_t=X_t\otimes Y_t\in H_A\otimes H_B.
\]
For a time window \(\Delta\), define the time-averaged tensor amplitude
\[
C_\Delta
=
\frac{1}{|\Delta|}
\int_\Delta Z_t\,dt
=
\frac{1}{|\Delta|}
\int_\Delta X_t\otimes Y_t\,dt
\in H_A\otimes H_B.
\]
This is a random vector \(C_\Delta=C_\Delta(\omega)\) in \(H_A\otimes H_B\).

The \emph{double covariance operator} is defined as the statistical covariance
of this random tensor:
\[
\widehat C
=
\mathbb E\bigl[\,|C_\Delta\rangle\langle C_\Delta|\,\bigr]
\in L(H_A\otimes H_B).
\]
The corresponding quantum state is obtained by trace normalization:
\[
\rho_{C}
=
\frac{\widehat C}{\operatorname{Tr}\widehat C}.
\]

Using the identification \(H_A\otimes H_B\cong L(H_B,H_A)\), the random vector
\(C_\Delta\) may equivalently be viewed as the random operator
\[
\widehat C_\Delta
=
\frac{1}{|\Delta|}
\int_\Delta |X_t\rangle\langle Y_t|\,dt
\in L(H_B,H_A).
\]
Thus the DCM may be interpreted either as the covariance of the random tensor
\(C_\Delta\in H_A\otimes H_B\) or as the covariance of the random operator
\(\widehat C_\Delta\in L(H_B,H_A)\).

Expanding the definition of \(\widehat C\), we obtain
\[
\widehat C
=
\frac{1}{|\Delta|^2}
\int_\Delta\int_\Delta
\hat K(t,s)\,dt\,ds,
\]
where
\[
\hat K(t,s)
=
\mathbb E\Bigl[
|X_t\otimes Y_t\rangle
\langle X_s\otimes Y_s|
\Bigr]
\in L(H_A\otimes H_B).
\]
The kernel \(\hat K(t,s)\) is a fourth-order object. The Gaussian reduction theorem
proved in the next section shows that for jointly Gaussian processes it is
completely determined by the second-order covariance kernels
\[
\hat R_{XX},\quad \hat R_{YY},\quad \hat R_{XY},\quad \hat R_{YX}.
\]
This reduction is the basic mathematical mechanism that allows one to describe
DCM states in terms of Gaussian covariance data alone.

\subsection{Second-order covariance kernels}

Define
\begin{equation}
\label{DO3}
\hat R_{XX}(t,s)
=
\mathbb E\big[|X_t\rangle\langle X_s|\big]
\in L(H_A),\; \hat  
R_{YY}(t,s)
=
\mathbb E\big[|Y_t\rangle\langle Y_s|\big]
\in L(H_B),
\end{equation}
\begin{equation}
\label{DO5}
\hat  R_{XY}(t,s)
=
\mathbb E\big[|X_t\rangle\langle Y_s|\big]
\in L(H_B,H_A),\; 
\hat R_{YX}(t,s)
=
\mathbb E\big[|Y_t\rangle\langle X_s|\big]
\in L(H_A,H_B).
\end{equation}
These covariances form the covariance operator of the joint Gaussian process
\begin{equation}
\label{DO7}
\hat  \Gamma(t,s)
=
\begin{pmatrix}
\hat  R_{XX}(t,s) & \hat R_{XY}(t,s)\\
\hat  R_{YX}(t,s) & \hat  R_{YY}(t,s)
\end{pmatrix}.
\end{equation}

\section{Wick reduction}

Let $u,u'\in H_A$ and $v,v'\in H_B$. Then
$$
\langle u\otimes v,
\hat K(t,s)
(u'\otimes v')
\rangle\\
=
\mathbb E\Big[
\langle u,X_t\rangle
\langle v,Y_t\rangle
\overline{\langle u',X_s\rangle}
\overline{\langle v',Y_s\rangle}
\Big].
$$

Since the underlying variables are jointly Gaussian and circular, Isserlis--Wick theorem yields (see appendix for details): 
$$
\mathbb E\Big[
\langle u,X_t\rangle
\langle v,Y_t\rangle
\overline{\langle u',X_s\rangle}
\overline{\langle v',Y_s\rangle}
\Big]
=
\langle u,\hat R_{XX}(t,s)u'\rangle
\langle v,\hat R_{YY}(t,s)v'\rangle\\
+
\langle u,\hat R_{XY}(t,s)v'\rangle
\langle v,\hat R_{YX}(t,s)u'\rangle.
$$
Define the exchange operator $\hat E(t,s)$ by
\begin{equation}
\label{DO8}
\langle u\otimes v,
\hat  E(t,s)
(u'\otimes v')
\rangle
=
\langle u, \hat R_{XY}(t,s)v'\rangle
\langle v,\hat R_{YX}(t,s)u'\rangle.
\end{equation}
Then
\begin{equation}
\label{DO9}
\hat K(t,s)
=
\hat R_{XX}(t,s)\otimes \hat R_{YY}(t,s)
+
\hat  E(t,s).
\end{equation}
Consequently,
\begin{equation}
\label{D=10}
\widehat C
=
\frac{1}{\Delta^2}
\int_\Delta\!\!\int_\Delta
\Big(
\hat R_{XX}(t,s)\otimes \hat R_{YY}(t,s)
+
\hat  E(t,s)
\Big)
dt\,ds.
\end{equation}

\subsection{Gaussian Reduction Theorem}

\begin{theorem}[Gaussian Reduction of DCM]
Let $X_t$ and $Y_t$ be jointly Gaussian Hilbert-space-valued stochastic
processes with finite second moments.

Then the DCM covariance operator
$
\widehat C
=
\mathbb E[|C_\Delta\rangle\langle C_\Delta|]
$
and therefore the normalized density operator
$\rho
=
\frac{\widehat C}{\operatorname{Tr}\widehat C}
$
are completely determined by the second-order covariance operator
\[
\hat \Gamma(t,s)
=
\begin{pmatrix}
\hat R_{XX}(t,s) & \hat R_{XY}(t,s)\\
\hat R_{YX}(t,s) & \hat R_{YY}(t,s)
\end{pmatrix}
\]
of the joint Gaussian process $(X_t,Y_t)$. Equivalently, all fourth-order moments entering the DCM construction
admit a Wick reduction to quadratic expressions in the covariance kernels,
and no independent fourth-order information contributes to 
the  double-covariance operator.
\end{theorem}

\begin{proof}
The operator $\widehat C$ depends on the kernel
\[
\hat K(t,s)
=
\mathbb E[|X_t\otimes Y_t\rangle\langle X_s\otimes Y_s|].
\]
For arbitrary vectors $u,u'\in H_A$ and $v,v'\in H_B$, matrix elements
of $\hat K(t,s)$ are fourth-order moments of jointly Gaussian scalar variables. 
Applying Isserlis--Wick theorem expresses each such fourth-order moment
uniquely in terms of pair covariances. Hence all matrix elements of
$\hat K(t,s)$ are determined by
$\hat R_{XX},\hat R_{YY},\hat R_{XY},\hat R_{YX}$. Therefore $\hat K(t,s)$ itself is determined by $\hat \Gamma(t,s)$.
Integrating over the observation window yields $\widehat C$, and
normalization yields $\rho$.
\end{proof}

\begin{remark}
For non-Gaussian processes the theorem generally fails. Higher-order
cumulants contribute to $K(t,s)$ and the DCM density operator may then
contain genuinely fourth-order statistical information.
\end{remark}

\section{Gaussian Realization of Arbitrary Pure States of Composite Quantum Systems}

We investigate the possibility of generating arbitrary finite-dimensional
pure quantum states of composite quantum systems  within the DCM framework by means of Gaussian
Hilbert-space-valued stochastic processes. The construction is based on a special organization of
Karhunen--Lo\`eve modes into {\it orthogonal temporal sectors} associated
with the Schmidt decomposition of the target state.

The resulting realization reveals an unexpected feature of DCM.
The generated quantum state is determined by the geometry of
time-averaged tensor modes rather than by ordinary second-order
cross-covariances. In particular, the construction remains valid even when the
underlying Gaussian processes are independent.

Previous DCM constructions \cite{DCM} demonstrated that arbitrary pure quantum
states of composite systems (e.g., the Bell states) can be generated from suitably organized stochastic processes.
Those constructions relied on explicitly synchronized temporal
structures and were generally non-Gaussian. The purpose of the present work is to investigate whether arbitrary
pure states can also be generated within the class of Gaussian
Hilbert-space-valued stochastic processes.

\subsection{The DCM State Construction}

The following elementary lemma will be used repeatedly.

\begin{lemma}
\label{Lemma1} Assume that
$
C_\Delta(\omega)=
\alpha(\omega)|\Psi \rangle \; (a.e.),
$
where $\alpha$ is a complex valued random variable with finite second moment and vector 
$\Psi \neq 0$ is fixed.
Then
$\hat \rho_{C}
= \frac{
|\Psi\rangle\langle\Psi|
}
{
\|\Psi\|^2
}.
$
In particular, if $\|\Psi\|=1$, then
$
\rho_{C}
=
|\Psi\rangle\langle\Psi|.
$
\end{lemma}

\begin{proof}
Since
$|C_\Delta\rangle\langle C_\Delta|
=
|\alpha|^2
|\Psi\rangle\langle\Psi|,
$
one obtains
$
E\!\left[
|C_\Delta\rangle\langle C_\Delta|
\right]
=
E\!\left[
|\alpha|^2
\right]
|\Psi\rangle\langle\Psi|.
$
Furthermore,
 $
\|C_\Delta\|^2
=
|\alpha|^2
\|\Psi\|^2,
$
hence
$
E\!\left[
\|C_\Delta\|^2
\right]
=
E\!\left[
|\alpha|^2
\right]
\|\Psi\|^2.
$
Substitution into the definition of $\rho_{C}$
gives the result.
\end{proof}

\subsection{Construction of Gaussian Processes}

Let
\[
\{u_k\}_{k\ge1}\subset {\cal H}_A\equiv L^2(\Delta;H_A), \;\{v_\ell\}_{\ell\ge1}\subset {\cal H}_B\equiv L^2(\Delta;H_B)
\]
be orthonormal families, and let $\{\lambda_k\}_{k\ge1},
\qquad
\{\mu_\ell\}_{\ell\ge1}
$
be positive summable sequences;
\[
\sum_{k=1}^{\infty}\lambda_k<\infty,\; 
\sum_{\ell=1}^{\infty}\mu_\ell<\infty.
\]
Let
\(\{\xi_k\}_{k\ge1}\)
and
\(\{\eta_\ell\}_{\ell\ge1}\)
be jointly Gaussian standard random variables with a prescribed covariance structure.

Define the Hilbert-space-valued stochastic processes
\[
X_t
=
\sum_{k=1}^{\infty}
\sqrt{\lambda_k}\,
\xi_k\,u_k(t),\; 
Y_t
=
\sum_{\ell=1}^{\infty}
\sqrt{\mu_\ell}\,
\eta_\ell\,v_\ell(t).
\]
These series converge in
$
L^2\!\bigl(\Omega;{\cal H}_A\bigr)$ and 
$
L^2\!\bigl(\Omega;{\cal H}_B\bigr),
$
respectively. Consequently, \(X_t\) and \(Y_t\) are centered Gaussian Hilbert-space-valued stochastic processes.

Their covariance kernels are the operator-valued kernels
$$
\hat R_{XX}(t,s)
=
E\!\left[
|X_t\rangle\langle X_s|
\right],\; 
\hat R_{YY}(t,s)
=
E\!\left[
|Y_t\rangle\langle Y_s|
\right].
$$
A direct computation yields
$$
\hat R_{XX}(t,s)
=
\sum_{k=1}^{\infty}
\lambda_k\,
|u_k(t)\rangle\langle u_k(s)|,\; 
\hat R_{YY}(t,s)
=
\sum_{\ell=1}^{\infty}
\mu_\ell\,
|v_\ell(t)\rangle\langle v_\ell(s)|.
$$

The corresponding covariance operators on ${\cal H}_A$
 ${\cal H}_B$
are positive and trace class. Their eigenfunctions are
$\{u_k\}$
and
$\{v_\ell\}$,
with corresponding eigenvalues
$\{\lambda_k\}$
and
$\{\mu_\ell\}$.

Substituting the above expansions into DCM gives  
\[
C_\Delta
=
\sum_{k,\ell}
\sqrt{\lambda_k\mu_\ell}\,
\xi_k\eta_\ell\,
W_{k\ell}
\]
where
\[
W_{k\ell}
=
\frac{1}{\Delta}
\int_\Delta
u_k(t)\otimes v_\ell(t)\,dt.
\]
Thus the DCM temporal covariance, and therefore the resulting DCM state, is completely determined by the deterministic tensor family
$\{W_{k\ell}\}.$

\subsection{Temporal Schmidt-Sector Construction}

Let
\[
|\Psi\rangle
=
\sum_{r=1}^{m}
s_r\, e_r \otimes f_r
\]
be a normalized Schmidt decomposition of a pure state,
$s_r >0, \sum_{r=1}^{m} s_r^2 =1.$

Define
\[
S
=
\sum_{r=1}^{m}
s_r .
\]

Denote by $L^2(\Delta)$ space of square integrable complex-valued functions defined on the segment $\Delta.$ Consider its decomposition into direct sum of subspaces: 

\[
L^2(\Delta)
=
F_1
\oplus
F_2
\oplus
\cdots
\oplus
F_m ,
\]
where each subspace \(F_r\) is infinite-dimensional.

For every \(r\), choose an orthonormal basis
$\{f_r^{(k)}\}_{k\ge 1}$
 in the subspace $F_r.$ Thus
\[
\int_\Delta
f_r^{(k)}(t)
f_q^{(\ell)}(t)
\,dt
=
\delta_{rq}
\delta_{k\ell}.
\]

Define

\[
u_k(t)
=
\frac{1}{\sqrt S}
\sum_{r=1}^{m}
\sqrt{s_r}
\,f_r^{(k)}(t)
\,e_r ,\;  v_k(t)
=
\frac{1}{\sqrt S}
\sum_{r=1}^{m}
\sqrt{s_r}
\,f_r^{(k)}(t)
\,f_r .
\]

\begin{lemma}
The families
\[
\{u_k\}_{k\ge 1}
\subset
{\cal H}_A, \; \; \{v_k\}_{k\ge 1}
\subset
{\cal H}_B
\]
are orthonormal.
\end{lemma}

\begin{proof}
Using the orthogonality relations,
$$
\langle u_k,u_\ell\rangle
=
\frac{1}{S}
\sum_{r,q}
\sqrt{s_r s_q}
\,
\langle f_r^{(k)},f_q^{(\ell)}\rangle
\,
\langle e_r,e_q\rangle \\
=
\frac{1}{S}
\sum_r
s_r
\delta_{k\ell} \\
=
\delta_{k\ell}.
$$
The proof for the family \(\{v_k\}\) is identical.
\end{proof}

\paragraph{Coordinate representation of the Gaussian processes.}

Let
\[
\{\lambda_k\}_{k\ge1},
\; 
\{\mu_k\}_{k\ge1},
\]
be positive summable sequences and let
\[
\{\xi_k\}_{k\ge1},\; 
\{\eta_k\}_{k\ge1},
\]
be independent standard circular Gaussian random variables.

The Gaussian processes generated by the above orthonormal systems are

\[
X_t
=
\sum_{k=1}^{\infty}
\sqrt{\lambda_k}\,
\xi_k\,u_k(t),
\]

and

\[
Y_t
=
\sum_{k=1}^{\infty}
\sqrt{\mu_k}\,
\eta_k\,v_k(t).
\]

Substituting the definitions of \(u_k\) and \(v_k\) gives

\[
\begin{aligned}
X_t
&=
\frac1{\sqrt S}
\sum_{r=1}^{m}
\sqrt{s_r}
\left(
\sum_{k=1}^{\infty}
\sqrt{\lambda_k}\,
\xi_k
\,f_r^{(k)}(t)
\right)e_r,
\end{aligned}
\]

and

\[
\begin{aligned}
Y_t
&=
\frac1{\sqrt S}
\sum_{r=1}^{m}
\sqrt{s_r}
\left(
\sum_{k=1}^{\infty}
\sqrt{\mu_k}\,
\eta_k
\,f_r^{(k)}(t)
\right)f_r.
\end{aligned}
\]

Thus each Hilbert-space component of the processes occupies its own temporal
sector \(F_r\). The Schmidt coefficients \(s_r\) determine the relative weights
of the temporal sectors, while the random Fourier coefficients are generated by
the Gaussian variables \(\{\xi_k\}\) and \(\{\eta_k\}\).

```latex
\subsection{Expected Energies of the Gaussian Processes}
\label{energy}

The coordinate representation of the processes immediately allows one to
compute their instantaneous and expected energies.

Since

\[
X_t
=
\frac1{\sqrt S}
\sum_{r=1}^{m}
\sqrt{s_r}
\left(
\sum_{k=1}^{\infty}
\sqrt{\lambda_k}\,
\xi_k
f_r^{(k)}(t)
\right)e_r,
\]

and the vectors
\(\{e_r\}_{r=1}^{m}\)
form an orthonormal basis of \(H_A\),

\[
\begin{aligned}
\|X_t\|^2
&=
\left\langle X_t,X_t\right\rangle
\\
&=
\frac1S
\sum_{r=1}^{m}
s_r
\left|
\sum_{k=1}^{\infty}
\sqrt{\lambda_k}\,
\xi_k
f_r^{(k)}(t)
\right|^2.
\end{aligned}
\]

Similarly,

\[
\|Y_t\|^2
=
\frac1S
\sum_{r=1}^{m}
s_r
\left|
\sum_{k=1}^{\infty}
\sqrt{\mu_k}\,
\eta_k
f_r^{(k)}(t)
\right|^2.
\]

These quantities are random and describe the instantaneous energies of the two
Gaussian processes.

Taking expectations and using

\[
E(\xi_k\overline{\xi_\ell})
=
\delta_{k\ell},
\]

we obtain

\[
\begin{aligned}
E\|X_t\|^2
&=
\frac1S
\sum_{r=1}^{m}
s_r
E
\left|
\sum_{k=1}^{\infty}
\sqrt{\lambda_k}\,
\xi_k
f_r^{(k)}(t)
\right|^2
\\
&=
\frac1S
\sum_{r=1}^{m}
s_r
\sum_{k=1}^{\infty}
\lambda_k
|f_r^{(k)}(t)|^2.
\end{aligned}
\]

Analogously,

\[
E\|Y_t\|^2
=
\frac1S
\sum_{r=1}^{m}
s_r
\sum_{k=1}^{\infty}
\mu_k
|f_r^{(k)}(t)|^2.
\]

Integrating over the interval \(\Delta\) and using the orthonormality of the
systems \(\{f_r^{(k)}\}_{k\ge1}\),

\[
\int_\Delta
|f_r^{(k)}(t)|^2\,dt
=
1,
\]

yields

\[
\begin{aligned}
\int_\Delta
E\|X_t\|^2\,dt
&=
\frac1S
\sum_{r=1}^{m}
s_r
\sum_{k=1}^{\infty}
\lambda_k
\\
&=
\sum_{k=1}^{\infty}
\lambda_k,
\end{aligned}
\]

since

\[
\sum_{r=1}^{m}s_r=S.
\]

Similarly,

\[
\int_\Delta
E\|Y_t\|^2\,dt
=
\sum_{k=1}^{\infty}
\mu_k.
\]

Thus the total expected energies of the two Gaussian processes coincide with
the traces of their covariance operators,

\[
\boxed{
\int_\Delta
E\|X_t\|^2\,dt
=
\operatorname{Tr}(Q_X)
=
\sum_{k=1}^{\infty}\lambda_k,
}
\]

and

\[
\boxed{
\int_\Delta
E\|Y_t\|^2\,dt
=
\operatorname{Tr}(Q_Y)
=
\sum_{k=1}^{\infty}\mu_k.
}
\]

Moreover, the contribution of the \(r\)-th temporal sector to the total
expected energy equals

\[
E_r(X)
=
\frac{s_r}{S}
\sum_{k=1}^{\infty}\lambda_k,
\]

and

\[
E_r(Y)
=
\frac{s_r}{S}
\sum_{k=1}^{\infty}\mu_k.
\]

Hence the relative expected energy carried by the \(r\)-th temporal sector is

\[
\frac{E_r(X)}
{\sum_{q=1}^{m}E_q(X)}
=
\frac{s_r}{S},
\qquad
\frac{E_r(Y)}
{\sum_{q=1}^{m}E_q(Y)}
=
\frac{s_r}{S}.
\]

Therefore the  Schmidt decomposition coefficients \(s_r\) admit a direct stochastic interpretation:
they determine how the expected energy of each Gaussian process is distributed
among the mutually orthogonal temporal sectors \(F_r\).

\begin{remark}
One may ask whether it is necessary to use the same temporal basis
\[
\{f_r^{(k)}\}_{k\ge1}\subset F_r
\]
in the constructions of both families
\(\{u_k\}\) and \(\{v_k\}\).
Suppose instead that the second family is constructed from another collection
of functions
\[
\{g_r^{(k)}\}_{k\ge1}\subset F_r.
\]

Repeating the computation of the tensors \(W_{k\ell}\) yields

\[
W_{k\ell}
=
\frac{1}{S|\Delta|}
\sum_{r,q=1}^{m}
\sqrt{s_rs_q}
\left(
\int_\Delta
f_r^{(k)}(t)\,
g_q^{(\ell)}(t)\,dt
\right)
e_r\otimes f_q.
\]

Hence the identity

\[
W_{k\ell}
=
\frac{\delta_{k\ell}}{S|\Delta|}
\sum_{r=1}^{m}
s_r\,e_r\otimes f_r
\]

holds provided the two systems satisfy the biorthogonality relations

\[
\int_\Delta
f_r^{(k)}(t)\,
g_q^{(\ell)}(t)\,dt
=
\delta_{rq}\delta_{k\ell}.
\]

Thus, from the viewpoint of the tensor realization alone, biorthogonality is
sufficient.

However, the present construction requires the families
\(\{u_k\}\) and \(\{v_k\}\) to be Karhunen--Lo\`eve modes of Gaussian processes.
Consequently, both
\(\{f_r^{(k)}\}\) and
\(\{g_r^{(k)}\}\)
must be complete orthonormal systems of the corresponding temporal sectors
\(F_r\).
A standard result from Hilbert space theory states that the biorthogonal basis
of a complete orthonormal basis coincides with the basis itself, up to
unimodular phase factors. Therefore,

\[
g_r^{(k)}(t)
=
e^{i\theta_r^{(k)}}\,f_r^{(k)}(t),
\]

for suitable real numbers \(\theta_r^{(k)}\).

Consequently, apart from inessential phase factors (which may be absorbed into
the Schmidt vectors), the use of the same temporal basis in both Gaussian
processes is essentially forced by the Karhunen--Lo\`eve structure. The apparent
symmetry of the construction is therefore not an additional assumption but a
consequence of simultaneously requiring tensor realization and orthonormal
Karhunen--Loève expansions.
\end{remark}

\section{The Fundamental Tensor Identity}

The key observation is that the integrated tensor modes
\(W_{k\ell}\) are all collinear with the target state.

\begin{proposition}
\label{Proposition1}

For the above construction,

\[
W_{k\ell}
=
\frac{\delta_{k\ell}}
     {S|\Delta|}
|\Psi \rangle.
\]

\end{proposition}

\begin{proof}

Using the definitions of \(u_k\) and \(v_\ell\),

\begin{align*}
u_k(t)\otimes v_\ell(t)
=
\frac{1}{S}
\sum_{r,q}
\sqrt{s_r s_q}
\,f_r^{(k)}(t)
\,f_q^{(\ell)}(t)
\,
e_r\otimes f_q .
\end{align*}

Therefore

\begin{align*}
W_{k\ell}
&=
\frac{1}{S|\Delta|}
\sum_{r,q}
\sqrt{s_r s_q}
\left(
\int_\Delta
f_r^{(k)}(t)
f_q^{(\ell)}(t)
\,dt
\right)
e_r\otimes f_q .
\end{align*}

Using

\[
\int_\Delta
f_r^{(k)}(t)
f_q^{(\ell)}(t)
\,dt
=
\delta_{rq}
\delta_{k\ell},
\]

we obtain

\begin{align*}
W_{k\ell}
&=
\frac{\delta_{k\ell}}
     {S|\Delta|}
\sum_r
s_r
\,e_r\otimes f_r.
\end{align*}

Since

\[
|\Psi \rangle
=
\sum_r
s_r
\,e_r\otimes f_r,
\]

the result follows.

\end{proof}

\section{Constructive Gaussian Realization}

In Section~8 we constructed centered circular Gaussian processes
\[
X_t
=
\sum_{k=1}^{\infty}
\sqrt{\lambda_k}\,
\xi_k\,u_k(t), \; Y_t
=
\sum_{\ell=1}^{\infty}
\sqrt{\mu_\ell}\,
\eta_\ell\,v_\ell(t),
\]
where
\[
\{u_k\}_{k\ge1}
\subset
{\cal H}_A,\; 
\{v_\ell\}_{\ell\ge1}
\subset
{\cal H}_B
\]
are orthonormal systems,
\[
\{\lambda_k\},\; 
\{\mu_\ell\}
\]
are positive summable sequences, and
\(\{\xi_k\}\), \(\{\eta_\ell\}\) are independent standard circular Gaussian random variables.

We first calculate the covariance operators of these processes. For a centered Hilbert-space-valued Gaussian process
$
X=\{X_t\}_{t\in\Delta},
$
its covariance operator
$
\hat T_X:{\cal H}_A\rightarrow {\cal H}_A
$
is defined by
\[
\hat T_X
=
E\bigl[\,|X\rangle\langle X|\,\bigr],
\]
that is,
\[
\langle \hat T_Xf,g\rangle
=
E\!\left[
\langle X,f\rangle
\overline{\langle X,g\rangle}
\right],
\qquad
f,g\in {\cal H}_A.
\]
Equivalently,
\[
(\hat T_Xf)(t)
=
\int_\Delta
\hat K_{XX}(t,s)f(s)\,ds,
\]
where
\[
\hat K_{XX}(t,s)
=
E\!\left[X_t\otimes X_s^{*}\right]
\]
is the covariance kernel.

Similarly one defines the covariance operator
\[
\hat T_Y
=
E\bigl[\,|Y\rangle\langle Y|\,\bigr]
\]
of the process \(Y\).

Using
\[
E(\xi_k\overline{\xi_\ell})
=
\delta_{k\ell},
\]
one immediately obtains
\[
\hat T_X
=
\sum_{k=1}^{\infty}
\lambda_k
\,|u_k\rangle\langle u_k|,
\]
that is,
\[
(\hat T_X f)(t)
=
\sum_{k=1}^{\infty}
\lambda_k
\langle f,u_k\rangle
u_k(t).
\]

Consequently,
\[
\hat K_{XX}(t,s)
=
\sum_{k=1}^{\infty}
\lambda_k
\,u_k(t)u_k(s)^*.
\]

Likewise,
\[
\hat T_Y
=
\sum_{\ell=1}^{\infty}
\mu_\ell
\,|v_\ell\rangle\langle v_\ell|,
\]
with covariance kernel
\[
\hat K_{YY}(t,s)
=
\sum_{\ell=1}^{\infty}
\mu_\ell
\,v_\ell(t)v_\ell(s)^*.
\]

Thus the covariance operators are diagonal in the orthonormal systems
\(\{u_k\}\) and \(\{v_\ell\}\), with eigenvalues
\(\{\lambda_k\}\) and \(\{\mu_\ell\}\), respectively.

In the previous sections, we constructed particular orthonormal systems associated with the Schmidt decomposition of a prescribed target state
$
|\Psi \rangle
=
\sum_{r=1}^{m}
s_r\,e_r\otimes f_r,
$
and proved that the corresponding deterministic tensors satisfy
\[
W_{k\ell}
=
\frac{1}{|\Delta|}
\int_\Delta
u_k(t)\otimes v_\ell(t)\,dt
=
\frac{\delta_{k\ell}}{S|\Delta|}
|\Psi \rangle.
\]
Substituting these tensors into the Gaussian reduction formula
\[
C_\Delta
=
\sum_{k,\ell}
\sqrt{\lambda_k\mu_\ell}\,
\xi_k\eta_\ell\,
W_{k\ell},
\]
yields
\[
C_\Delta
=
\frac{|\Psi \rangle}{S|\Delta|}
\sum_{k=1}^{\infty}
\sqrt{\lambda_k\mu_k}\,
\xi_k\eta_k.
\]

Define
\[
\alpha(\omega)
=
\frac{1}{S|\Delta|}
\sum_{k=1}^{\infty}
\sqrt{\lambda_k\mu_k}\,
\xi_k(\omega)\eta_k(\omega).
\]

Then
\[
C_\Delta(\omega)
=
\alpha(\omega)|\Psi \rangle.
\]

To apply Lemma \ref{Lemma1} it remains to verify that
\[
E|\alpha|^2<\infty.
\]

Since the Gaussian processes are independent,
\[
E\!\left[
\xi_k\eta_k
\overline{\xi_\ell\eta_\ell}
\right]
=
E[\xi_k\overline{\xi_\ell}]
\,
E[\eta_k\overline{\eta_\ell}]
=
\delta_{k\ell},
\]
and therefore
\[
\begin{aligned}
E|\alpha|^2
&=
\frac{1}{S^2|\Delta|^2}
\sum_{k,\ell}
\sqrt{\lambda_k\mu_k\lambda_\ell\mu_\ell}\,
E\!\left[
\xi_k\eta_k
\overline{\xi_\ell\eta_\ell}
\right] \\
&=
\frac{1}{S^2|\Delta|^2}
\sum_{k=1}^{\infty}
\lambda_k\mu_k.
\end{aligned}
\]

Since
\[
\sum_{k=1}^{\infty}\lambda_k<\infty,
\qquad
\sum_{k=1}^{\infty}\mu_k<\infty,
\]
both sequences are bounded. Hence
\[
\sum_{k=1}^{\infty}\lambda_k\mu_k
\le
\left(\sup_k\mu_k\right)
\sum_{k=1}^{\infty}\lambda_k
<
\infty,
\]
and consequently
\[
E|\alpha|^2<\infty.
\]
We have therefore established the following realization theorem.

\begin{theorem}[Constructive Gaussian Realization]
For every finite-dimensional pure state
\[
|\Psi \rangle
=
\sum_{r=1}^{m}
s_r\,e_r\otimes f_r,
\]
there exist centered circular Gaussian Hilbert-space-valued stochastic processes
\[
X_t:\Omega\rightarrow H_A,
\qquad
Y_t:\Omega\rightarrow H_B,
\]
such that
\[
\rho_\Delta
=
|\Psi \rangle\langle\Psi|.
\]
\end{theorem}

\begin{proof}
The preceding calculations show that
\[
C_\Delta(\omega)
=
\alpha(\omega)\Psi,
\]
where
\[
E|\alpha|^2<\infty.
\]
The conclusion therefore follows directly from Lemma~1.
\end{proof}

\begin{remark}
The Gaussian realization theorem reveals a structural feature of the DCM framework that differs substantially from the standard quantum-information intuition. In the present construction, one may choose the cross-covariance kernel to vanish identically,
\[
\hat R_{XY}(t,s)\equiv 0.
\]
For jointly Gaussian processes this implies that \(X_t\) and \(Y_t\) are independent.
Nevertheless, the resulting DCM state may be an arbitrary pure entangled state, including maximally entangled states such as Bell states. Thus the generation of entanglement in the DCM framework does not require nontrivial second-order cross-correlations between the underlying stochastic processes.

Instead, the entanglement is encoded in the temporal geometry of the deterministic tensors
$W_{k\ell}$ which enter the time-averaged amplitude
\[
C_\Delta
=
\frac{1}{|\Delta|}
\int_\Delta
X_t\otimes Y_t\,dt.
\]

The DCM construction therefore demonstrates that strong entanglement may emerge entirely from fourth-order temporal coherence, even when all ordinary second-order cross-covariances vanish.
\end{remark}

We can formulate this property of DCM  as mathematical statement:

\begin{corollary}[Entanglement without cross-covariance]
For every finite-dimensional pure state
\[
|\Psi \langle \in H_A\otimes H_B,
\]
there exist independent Gaussian processes \(X_t\) and \(Y_t\) such that the corresponding DCM state satisfies
\[
\hat \rho
=
|\Psi\rangle\langle\Psi|.
\]

In particular, there exist independent Gaussian processes whose DCM state is maximally entangled.
\end{corollary}

\begin{proof}
Choose
\[
\hat R_{XY}(t,s)\equiv 0.
\]
For jointly Gaussian processes this implies independence. By Theorem~\ref{thm:GaussianRealization}, the realization depends only on the deterministic tensor family
\[
W_{k\ell}
=
\frac{1}{|\Delta|}
\int_\Delta
u_k(t)\otimes v_\ell(t)\,dt
\]
and not on the cross-covariance kernel. Therefore the resulting DCM state remains
\[
\hat \rho
=
|\Psi\rangle\langle\Psi|.
\]
\end{proof}

\subsection{A Concrete Two-Dimensional Example}
\label{example}

We illustrate the realization theorem by an explicit construction.
Let
$\Delta=[0,2\pi],
$
and let
\[
H_A=\operatorname{span}\{e_1,e_2\},\; H_B=\operatorname{span}\{f_1,f_2\},
\]
where both bases are orthonormal.
We decompose the temporal Hilbert space as
$
L^2([0,2\pi])
=
F_1\oplus F_2,
$
where
\[
F_1
=
\overline{\operatorname{span}}
\left\{ \frac{1}{\sqrt{2\pi}},
\frac{\cos(kt)}{\sqrt{\pi}}
:
k\ge1
\right\},\;
F_2
=
\overline{\operatorname{span}}
\left\{
\frac{\sin(kt)}{\sqrt{\pi}}
:
k\ge1
\right\}.
\]
Let
\[
|\Psi \rangle
=
s_1e_1\otimes f_1
+
s_2e_2\otimes f_2, \; s_1,s_2>0,
s_1^2+s_2^2=1,
\]
and define
$S=s_1+s_2.$
The Karhunen--Lo\`eve modes are chosen as
\[
u_k(t)
=
\frac1{\sqrt S}
\left(
\sqrt{s_1}
\frac{\cos(kt)}{\sqrt{\pi}}
e_1
+
\sqrt{s_2}
\frac{\sin(kt)}{\sqrt{\pi}}
e_2
\right),
\;
v_k(t)
=
\frac1{\sqrt S}
\left(
\sqrt{s_1}
\frac{\cos(kt)}{\sqrt{\pi}}
f_1
+
\sqrt{s_2}
\frac{\sin(kt)}{\sqrt{\pi}}
f_2
\right).
\]
The orthogonality of the trigonometric system implies
that both families are orthonormal.

Let
$
\{\xi_k\}_{k\ge1}, \{\eta_k\}_{k\ge1},
$
be independent standard circular Gaussian random variables and let
$\{\lambda_k\},
\qquad
\{\mu_k\},$
be positive summable sequences. The corresponding Gaussian processes are presented in the coordinate form as
\begin{equation}
\label{Xt}
\begin{aligned}
X_t
&=
\frac{\sqrt{s_1}}{\sqrt{\pi S}}
\left(
\sum_{k=1}^{\infty}
\sqrt{\lambda_k}\,
\xi_k
\cos(kt)
\right)e_1
+
\frac{\sqrt{s_2}}{\sqrt{\pi S}}
\left(
\sum_{k=1}^{\infty}
\sqrt{\lambda_k}\,
\xi_k
\sin(kt)
\right)e_2,
\end{aligned}
\end{equation}
and
\begin{equation}
\label{yt}
\begin{aligned}
Y_t
&=
\frac{\sqrt{s_1}}{\sqrt{\pi S}}
\left(
\sum_{k=1}^{\infty}
\sqrt{\mu_k}\,
\eta_k
\cos(kt)
\right)f_1
+
\frac{\sqrt{s_2}}{\sqrt{\pi S}}
\left(
\sum_{k=1}^{\infty}
\sqrt{\mu_k}\,
\eta_k
\sin(kt)
\right)f_2.
\end{aligned}
\end{equation}

Thus the first Hilbert-space component of each process occupies the cosine
sector \(F_1\), while the second occupies the sine sector \(F_2\).
The deterministic tensors are
$W_{k\ell}
=
\frac1{2\pi}
\int_0^{2\pi}
u_k(t)\otimes v_\ell(t)\,dt.
$
Substituting the expressions for \(u_k\) and \(v_\ell\) gives
\[
W_{k\ell}
=
\frac{\delta_{k\ell}}{2\pi S}
\left(
s_1e_1\otimes f_1
+
s_2e_2\otimes f_2
\right)
=
\frac{\delta_{k\ell}}{2\pi S}|\Psi \rangle.
\]

\section{Concurrence and temporal energy distribution}

We start with the two qubit case considered in section \ref{example}. To be more concrete we continue the example presented in setcion \ref{example}.

The expected energies computed above show that, for the present realization,
the cosine and sine temporal sectors carry relative energies
\[
p_1:=\frac{E_1(X)}{E_1(X)+E_2(X)}=\frac{s_1}{s_1+s_2},
\qquad
p_2:=\frac{E_2(X)}{E_1(X)+E_2(X)}=\frac{s_2}{s_1+s_2},
\]
with \(p_1+p_2=1\). The same relations hold for the process \(Y\).

Thus the realization theorem associates to the Schmidt coefficients
\((s_1,s_2)\) a temporal energy distribution \((p_1,p_2)\) given by
\[
s_1=\frac{p_1}{\sqrt{p_1^2+p_2^2}},
\qquad
s_2=\frac{p_2}{\sqrt{p_1^2+p_2^2}}.
\]
Substituting into the concurrence formula for a pure two-qubit state,
\[
\mathcal C(\Psi )=2s_1s_2,
\]
we obtain
\[
\mathcal C(\Psi)
=
\frac{2p_1p_2}{p_1^2+p_2^2}.
\]

Hence, within the present DCM realization, concurrence is determined by the
distribution of expected stochastic energy between the two temporal sectors.
It vanishes when one sector carries all the energy, and it is maximal when the
energy is equally distributed:
\[
p_1=p_2=\frac12
\quad\Longrightarrow\quad
\mathcal C(\Psi)=1.
\]

Therefore, in this construction, entanglement is controlled not directly by the
quadratic weights \(s_1^2,s_2^2\), but by the relative sector energies
\[
p_r=\frac{E_r(X)}{E_1(X)+E_2(X)}
=\frac{E_r(Y)}{E_1(Y)+E_2(Y)},
\qquad r=1,2.
\]
The Schmidt coefficients are recovered from these energy fractions by the above
normalization formulas.

Now we consider the general multidimensional case. 
Let
\[
|\Psi \rangle=\sum_{r=1}^m s_r\,e_r\otimes f_r,
\qquad s_r\ge 0,
\qquad \sum_{r=1}^m s_r^2=1,
\]
be the Schmidt decomposition of the target pure state.
A standard concurrence-type entanglement measure for pure bipartite states is
\[
\mathcal C(\Psi)
=
\sqrt{2\bigl(1-\operatorname{Tr} \hat \rho_A^2\bigr)},
\]
where
\[
\hat \rho_A=\operatorname{Tr}_{H_B}|\Psi\rangle\langle\Psi|
      =\sum_{r=1}^m s_r^2\,|e_r\rangle\langle e_r|.
\]
Hence
\[
\mathcal C(\Psi)
=
\sqrt{2\left(1-\sum_{r=1}^m s_r^4\right)}.
\]

For the Gaussian realization constructed above, the expected energies of the temporal sectors satisfy
\[
E_r(X)\propto s_r,
\qquad
E_r(Y)\propto s_r,
\qquad r=1,\dots,m.
\]
Therefore the normalized sector energies are
\[
p_r
:=
\frac{E_r(X)}{\sum_{q=1}^m E_q(X)}
=
\frac{E_r(Y)}{\sum_{q=1}^m E_q(Y)}
=
\frac{s_r}{\sum_{q=1}^m s_q}.
\]
Writing
\[
S:=\sum_{q=1}^m s_q,
\]
we have \(s_r=Sp_r\). Since \(\sum_{r=1}^m s_r^2=1\), it follows that
\[
1=\sum_{r=1}^m s_r^2=S^2\sum_{r=1}^m p_r^2,
\]
and therefore
\[
S=\frac{1}{\sqrt{\sum_{r=1}^m p_r^2}}.
\]
Consequently,
\[
s_r=\frac{p_r}{\sqrt{\sum_{q=1}^m p_q^2}},
\qquad r=1,\dots,m.
\]

Substituting into the concurrence formula yields
\[
\mathcal C(\Psi)
=
\sqrt{
2\left(
1-
\frac{\sum_{r=1}^m p_r^4}
     {\left(\sum_{q=1}^m p_q^2\right)^2}
\right)
}.
\]

Thus, in the present DCM realization, the entanglement of the target pure state
is completely determined by the distribution of expected stochastic energy among
the temporal sectors. The Schmidt coefficients are recovered from the sector
energy fractions \(\{p_r\}_{r=1}^m\) by the normalization formula
\[
s_r=\frac{p_r}{\sqrt{\sum_{q=1}^m p_q^2}},
\]
and the concurrence is then given by the above expression.

In the two-dimensional case \(m=2\), this reduces to
\[
\mathcal C(\Psi)
=
\frac{2p_1p_2}{p_1^2+p_2^2},
\]
which coincides with the formula obtained above.

\section{Recovery from the Gaussian Reduction Theorem}

The realization theorem may also be derived directly from the
Gaussian Reduction Theorem established in Section~\ref{sec:gaussian}.
Consider the special case $
\hat R_{XY}(t,s)=0.$
Then the exchange contribution
$
\mathcal E(t,s)
$
vanishes identically and the two-time kernel reduces to
\[
\hat K(t,s)
=
\hat R_{XX}(t,s)\otimes \hat R_{YY}(t,s).
\]
Using the expansions introduced above,
$
\hat R_{XX}(t,s)
=
\sum_k
\lambda_k
\,|u_k(t)\rangle
\langle u_k(s)|,
$
and
$
\hat R_{YY}(t,s)
=
\sum_\ell
\mu_\ell
\,|v_\ell(t)\rangle
\langle v_\ell(s)|,
$
we obtain

\[
\hat K(t,s)
=
\sum_{k,\ell}
\lambda_k\mu_\ell
\,
|u_k(t)\otimes v_\ell(t)\rangle
\langle
u_k(s)\otimes v_\ell(s)
|.
\]
Integrating over the observation interval and using the definition
of the DCM density operator yields
$
C
=
\sum_{k,\ell}
\lambda_k\mu_\ell
\,|W_{k\ell}\rangle
\langle W_{k\ell}|,
$
where
$
W_{k\ell}
=
\frac{1}{|\Delta|}
\int_\Delta
u_k(t)\otimes v_\ell(t)
\,dt.
$
By Proposition \ref{Proposition1},
$
W_{k\ell}
=
\frac{\delta_{k\ell}}
     {S|\Delta|}
|\Psi \rangle .
$
Substitution gives
$
C
=
\gamma
\,|\Psi\rangle
\langle\Psi|.$

Thus the Gaussian realization theorem follows directly from the
Gaussian Reduction Theorem.

Surprisingly  the exchange term \(\mathcal E\) is not required for
the generation of entangled states. In the present construction,
entanglement is generated entirely by the temporal geometry encoded
in the tensors \(W_{k\ell}\).

\section{Towards a General Theory of Gaussian DCM States}

The Gaussian realization theorem established in the previous section corresponds to a highly special
situation.
The temporal modes are arranged so that the deterministic tensors
\[
W_{k\ell}
=
\frac{1}{|\Delta|}
\int_\Delta
u_k(t)\otimes v_\ell(t)\,dt
\]
satisfy
\[
W_{k\ell}
=
\delta_{k\ell}\,c\,|\Psi \rangle.
\]
Consequently,
\[
C_\Delta(\omega)
=
\alpha(\omega)|\Psi \rangle,
\]
and the resulting DCM state is pure.

The Gaussian Reduction Theorem suggests a much broader viewpoint.
For independent Gaussian processes,
\[
C
=
\sum_{k,\ell}
\lambda_k\mu_\ell
|W_{k\ell}\rangle\langle W_{k\ell}|,
\]
so that the resulting DCM state is completely determined by the family
\[
\mathcal W
=
\{W_{k\ell}\}.
\]

Thus the fundamental object of Gaussian DCM theory is not the covariance kernel itself but rather
the geometry of the temporal tensor family.

This naturally leads to the following realizability problem.

\medskip

\noindent
\textbf{Realizability Problem.}
Characterize all tensor families
\[
\{W_{k\ell}\}\subset H_A\otimes H_B
\]
that admit a representation
\[
W_{k\ell}
=
\frac1{|\Delta|}
\int_\Delta
u_k(t)\otimes v_\ell(t)\,dt,
\]
where
\[
\{u_k\}\subset {\cal H}_A,
\qquad
\{v_\ell\}\subset {\cal H}_B
\]
are orthonormal systems.

The realization theorem provides one highly nontrivial example of such a family.
Different temporal organizations produce different tensor geometries and therefore different Gaussian
DCM states.

Several elementary structural observations follow immediately.

\paragraph{Pure states.}
The operator
\[
C
=
\sum_{k,\ell}
\lambda_k\mu_\ell
|W_{k\ell}\rangle\langle W_{k\ell}|
\]
has rank one if and only if all nonzero tensors are collinear,
\[
W_{k\ell}
=
c_{k\ell}|\Psi \rangle.
\]
After normalization,
\[
\hat \rho
=
\frac{|\Psi\rangle\langle\Psi|}
{\|\Psi\|^2}.
\]
The realization theorem established above is precisely of this type.

\paragraph{Separable states.}
Suppose every tensor factors,
\[
W_{k\ell}
=
a_{k\ell}\otimes b_{k\ell}.
\]
Then
\[
|W_{k\ell}\rangle\langle W_{k\ell}|
=
|a_{k\ell}\rangle\langle a_{k\ell}|
\otimes
|b_{k\ell}\rangle\langle b_{k\ell}|,
\]
and therefore
\[
\widehat C
=
\sum_{k,\ell}
\lambda_k\mu_\ell
\left(
|a_{k\ell}\rangle\langle a_{k\ell}|
\right)
\otimes
\left(
|b_{k\ell}\rangle\langle b_{k\ell}|
\right).
\]
Hence the normalized DCM state is separable.

Outside these two extreme situations, the geometry of the family
\(\{W_{k\ell}\}\) becomes essential.
The remainder of this section develops a general framework for analyzing entanglement through the
structure of these temporal tensors.

\subsection{Common Schmidt Decompositions}

A particularly tractable situation occurs when all tensors \(W_{k\ell}\) share a common Schmidt basis.

\begin{definition}
The family \(\{W_{k\ell}\}\) is said to admit a common Schmidt decomposition if there exist orthonormal systems
\[
\{e_r\}_{r=1}^{R}\subset H_A,
\qquad
\{f_r\}_{r=1}^{R}\subset H_B,
\]
such that
\[
W_{k\ell}
=
\sum_{r=1}^{R}
a_r^{(k\ell)}
\,e_r\otimes f_r
\]
for all pairs \((k,\ell)\).
\end{definition}

The common Schmidt decomposition condition for the family \(\{W_{k\ell}\}\)
may also be reformulated in operator terms by introducing the associated
Hilbert--Schmidt operators
\[
\hat M_{k\ell}:=\frac1{|\Delta|}\int_\Delta |u_k(t)\rangle\langle v_\ell(t)|\,dt.
\]
It has the form:
\[
\hat M_{k\ell}
=
\sum_{r=1}^{R}
a_r^{(k\ell)}
\,|e_r\rangle\langle f_r|.
\]

The coefficients are
\[
a_r^{(k\ell)}
=
\langle e_r,\hat M_{k\ell}f_r\rangle
=
\frac{1}{|\Delta|}
\int_\Delta
\langle e_r,u_k(t)\rangle
\langle v_\ell(t),f_r\rangle
\,dt.
\]

Thus the Schmidt coefficients are obtained directly from temporal overlaps of the mode functions.

Under the standard tensor--operator correspondence, the family
\(\{W_{k\ell}\}\subset H_A\otimes H_B\) admits a common Schmidt decomposition
if and only if the operators \(\{\hat M_{k\ell}\}\) admit a simultaneous singular-value
decomposition. We shall not pursue this operator-theoretic characterization here,
since for the Gaussian constructions considered below a more transparent sufficient
condition is given by the temporal sector decomposition introduced in
Proposition~\ref{Tsector}.

\subsection{The Schmidt-Coherence Matrix}

Assume from now on that a common Schmidt decomposition exists.

Define
\[
\Gamma_{rs}
=
\sum_{k,\ell}
\lambda_k\mu_\ell
\,a_r^{(k\ell)}
\overline{a_s^{(k\ell)}}.
\]

The matrix
\[
\Gamma
=
(\Gamma_{rs})_{r,s}
\]
will be called the Schmidt-coherence matrix.

It is positive semidefinite. Indeed,
\[
\Gamma
=
AA^{*},
\]
where the columns of \(A\) are
\[
\sqrt{\lambda_k\mu_\ell}
\left(
a_1^{(k\ell)},
\dots,
a_R^{(k\ell)}
\right)^T.
\]

Substituting the Schmidt expansions into the Gaussian reduction formula yields
\[
\widehat C
=
\sum_{r,s}
\Gamma_{rs}
\,|e_r\otimes f_r\rangle
\langle e_s\otimes f_s|.
\]

Thus the DCM covariance operator is completely determined by the Schmidt-coherence matrix.

\subsection{Entanglement Criterion}

We now derive a sufficient condition for entanglement.

\begin{proposition}[Schmidt-coherence criterion]
\label{cr}
Assume that the family \(\{W_{k\ell}\}\) admits a common Schmidt decomposition.

If
\[
\Gamma_{rs}\neq0
\]
for some \(r\neq s\),
then the normalized DCM state
\[
\hat \rho_C
=
\frac{\widehat C}{\operatorname{Tr}(\widehat C)}
\]
is entangled.
\end{proposition}

\begin{proof}
Taking the partial transpose on the second subsystem gives
\[
\hat C^\Gamma
=
\sum_{r,s}
\Gamma_{rs}
\,|e_r\otimes f_s\rangle
\langle e_s\otimes f_r|.
\]

Fix \(r\neq s\). The subspace
\[
\mathcal H_{rs}
=
\operatorname{span}
\{
e_r\otimes f_s,
e_s\otimes f_r
\}
\]
is invariant under \(C^\Gamma\).

Restricted to this subspace,
\(\hat C^\Gamma\) is represented by
\[
\begin{pmatrix}
0 & \Gamma_{rs}\\
\overline{\Gamma_{rs}} & 0
\end{pmatrix}.
\]

Its eigenvalues are
\[
\pm |\Gamma_{rs}|.
\]

Hence, if
\[
\Gamma_{rs}\neq0,
\]
the operator \(C^\Gamma\) possesses a negative eigenvalue and therefore is not positive.

Since positivity of the partial transpose is necessary for separability, the state
$
\hat \rho_C
=
\frac{\widehat C}{\operatorname{Tr}(\widehat C)}
$
cannot be separable.
\end{proof}

\subsection{Relation to the realization theorem}

The realization theorem obtained earlier corresponds to a particularly simple
instance of the common Schmidt decomposition framework. Let
\[
|\Psi \rangle=\sum_{r=1}^m s_r\,e_r\otimes f_r,
\qquad s_r\ge 0,\qquad \sum_{r=1}^m s_r^2=1,
\]
and let \(X,Y\) be the Gaussian processes constructed in the realization theorem.
Then the associated temporal tensors satisfy
\[
W_{k\ell}
=
\delta_{k\ell}\,\frac{1}{S|\Delta|}\,|\Psi \rangle,
\qquad
S:=\sum_{r=1}^m s_r.
\]
Equivalently,
\[
W_{k\ell}
=
\delta_{k\ell}\,c\sum_{r=1}^m s_r\,e_r\otimes f_r,
\qquad
c:=\frac{1}{S|\Delta|}.
\]
Thus the family \(\{W_{k\ell}\}\) has a common Schmidt decomposition with
coefficients
\[
a_r^{(k\ell)}=\delta_{k\ell}\,c\,s_r.
\]

The corresponding Schmidt-coherence matrix is therefore
\[
\Gamma_{rs}
=
\sum_{k,\ell}\lambda_k\mu_\ell\,
a_r^{(k\ell)}\overline{a_s^{(k\ell)}}
=
|c|^2\sum_k\lambda_k\mu_k\,s_rs_s.
\]
Hence
\[
\Gamma_{rs}=\beta\,s_rs_s,
\qquad
\beta:=|c|^2\sum_k\lambda_k\mu_k.
\]
In particular, \(\Gamma\) has rank one. Substituting into the general formula
for \(C\), we obtain
\[
C
=
\sum_{r,s}\Gamma_{rs}\,
|e_r\otimes f_r\rangle\langle e_s\otimes f_s|
=
\beta\sum_{r,s}s_rs_s\,
|e_r\otimes f_r\rangle\langle e_s\otimes f_s|
=
\beta\,|\Psi\rangle\langle\Psi|.
\]
After normalization, this yields
\[
\rho_C=|\Psi\rangle\langle\Psi|,
\]
which is exactly the conclusion of the realization theorem.

\subsection{Open Problems}

The preceding discussion suggests several natural problems.

\begin{enumerate}
\item Characterize all tensor families
\[
\{W_{k\ell}\}
\]
that admit a common Schmidt decomposition.

\item Determine temporal conditions on the Gaussian modes
\[
u_k(t),
\qquad
v_\ell(t),
\]
that imply simultaneous singular-value diagonalizability of the operators \(\hat M_{k\ell}\).

\item Develop entanglement criteria that do not require a common Schmidt decomposition.

\item Characterize all tensor families realizable in the form
\[
W_{k\ell}
=
\frac{1}{|\Delta|}
\int_\Delta
u_k(t)\otimes v_\ell(t)\,dt.
\]
\end{enumerate}

These problems point toward a general geometric theory of Gaussian DCM states beyond the pure-state realization theorem established in the present work.

\subsection{A Sufficient Condition for a Common Schmidt Decomposition}

The common Schmidt decomposition assumption admits a natural sufficient condition
expressed directly in terms of the temporal mode functions. The basic idea is to
assign different Schmidt directions to pairwise orthogonal temporal sectors.
Then all mixed sector contributions vanish after time integration, and the
resulting family \(\{W_{k\ell}\}\) acquires a common Schmidt decomposition.

\begin{proposition}[Temporal sector decomposition]
\label{Tsector}
Assume that the temporal Hilbert space admits an orthogonal decomposition
\[
L^2(\Delta)
=
\left(\bigoplus_{r=1}^{R} F_r\right)\oplus F_\perp,
\]
where \(F_1,\dots,F_R\subset L^2(\Delta)\) are pairwise orthogonal closed
subspaces and \(F_\perp\) is their orthogonal complement.

Let
\[
\{e_r\}_{r=1}^{R}\subset H_A,
\qquad
\{f_r\}_{r=1}^{R}\subset H_B
\]
be orthonormal systems, and suppose that the Gaussian modes admit representations
\[
u_k(t)=\sum_{r=1}^{R}\phi_{kr}(t)e_r,
\qquad
v_\ell(t)=\sum_{r=1}^{R}\psi_{\ell r}(t)f_r,
\]
where
\[
\phi_{kr}\in F_r,
\qquad
\psi_{\ell r}\in F_r
\qquad
(r=1,\dots,R).
\]

Then the family
\[
W_{k\ell}
=
\frac{1}{|\Delta|}
\int_\Delta
u_k(t)\otimes v_\ell(t)\,dt
\]
admits a common Schmidt decomposition. More precisely,
\[
W_{k\ell}
=
\sum_{r=1}^{R}
a_r^{(k\ell)}\,e_r\otimes f_r,
\]
where
\[
a_r^{(k\ell)}
=
\frac{1}{|\Delta|}
\int_\Delta
\phi_{kr}(t)\overline{\psi_{\ell r}(t)}\,dt.
\]
\end{proposition}

\begin{proof}
Substituting the mode expansions gives
\[
u_k(t)\otimes v_\ell(t)
=
\sum_{r,s=1}^{R}
\phi_{kr}(t)\overline{\psi_{\ell s}(t)}\,e_r\otimes f_s,
\]
and therefore
\[
W_{k\ell}
=
\sum_{r,s=1}^{R}
\left(
\frac{1}{|\Delta|}
\int_\Delta
\phi_{kr}(t)\overline{\psi_{\ell s}(t)}\,dt
\right)e_r\otimes f_s.
\]

Since \(F_r\perp F_s\) for \(r\neq s\), and
\[
\phi_{kr}\in F_r,
\qquad
\psi_{\ell s}\in F_s,
\]
we have
\[
\int_\Delta
\phi_{kr}(t)\overline{\psi_{\ell s}(t)}\,dt
=
0
\qquad
(r\neq s).
\]
Hence all off-diagonal terms vanish and only the contributions with \(r=s\)
remain:
\[
W_{k\ell}
=
\sum_{r=1}^{R}
\left(
\frac{1}{|\Delta|}
\int_\Delta
\phi_{kr}(t)\overline{\psi_{\ell r}(t)}\,dt
\right)e_r\otimes f_r.
\]
This is exactly the claimed common Schmidt decomposition.
\end{proof}

Proposition~\ref{Tsector} shows that the temporal-sector decomposition reduces
the construction of the family \(\{W_{k\ell}\}\) to the construction of the scalar
coefficients
\[
a_r^{(k\ell)}
=
\frac{1}{|\Delta|}
\int_\Delta
\phi_{kr}(t)\overline{\psi_{\ell r}(t)}\,dt.
\]
Once the Schmidt directions
\[
e_1\otimes f_1,\dots,e_R\otimes f_R
\]
are fixed, the resulting DCM state is determined by these coefficients and hence
by the Schmidt--coherence matrix
\[
\Gamma_{rs}
=
\sum_{k,\ell}
\lambda_k\mu_\ell\,
a_r^{(k\ell)}\overline{a_s^{(k\ell)}}.
\]

A particularly transparent specialization is obtained when, inside each temporal
sector \(F_r\), all \(A\)-components are proportional to one common profile and
all \(B\)-components are proportional to another common profile in the same
sector.

\begin{corollary}[Paired sector profiles]
\label{PairedSectorProfiles}
Assume the hypotheses of Proposition~\ref{Tsector}. In addition, suppose that for
each \(r=1,\dots,R\) there exist functions
\[
\chi_r^A\in F_r,
\qquad
\chi_r^B\in F_r,
\]
and complex coefficients \(\alpha_{kr},\beta_{\ell r}\in\mathbb C\) such that
\[
\phi_{kr}(t)=\alpha_{kr}\chi_r^A(t),
\qquad
\psi_{\ell r}(t)=\beta_{\ell r}\chi_r^B(t).
\]
Define the sector overlap constants
\[
\eta_r
:=
\frac{1}{|\Delta|}
\int_\Delta
\chi_r^A(t)\overline{\chi_r^B(t)}\,dt.
\]
Then
\[
W_{k\ell}
=
\sum_{r=1}^{R}
\eta_r\,\alpha_{kr}\overline{\beta_{\ell r}}\,e_r\otimes f_r.
\]
Equivalently,
\[
a_r^{(k\ell)}=\eta_r\,\alpha_{kr}\overline{\beta_{\ell r}}.
\]
\end{corollary}

\begin{proof}
By Proposition~\ref{Tsector},
\[
a_r^{(k\ell)}
=
\frac{1}{|\Delta|}
\int_\Delta
\phi_{kr}(t)\overline{\psi_{\ell r}(t)}\,dt.
\]
Substituting
\[
\phi_{kr}(t)=\alpha_{kr}\chi_r^A(t),
\qquad
\psi_{\ell r}(t)=\beta_{\ell r}\chi_r^B(t),
\]
we obtain
\[
a_r^{(k\ell)}
=
\frac{1}{|\Delta|}
\int_\Delta
\alpha_{kr}\chi_r^A(t)\,
\overline{\beta_{\ell r}\chi_r^B(t)}\,dt
=
\alpha_{kr}\overline{\beta_{\ell r}}
\frac{1}{|\Delta|}
\int_\Delta
\chi_r^A(t)\overline{\chi_r^B(t)}\,dt.
\]
This is exactly
\[
a_r^{(k\ell)}=\eta_r\,\alpha_{kr}\overline{\beta_{\ell r}}.
\]
Therefore
\[
W_{k\ell}
=
\sum_{r=1}^{R}
\eta_r\,\alpha_{kr}\overline{\beta_{\ell r}}\,e_r\otimes f_r.
\]
\end{proof}

Corollary~\ref{PairedSectorProfiles} shows that the coefficients of the common
Schmidt decomposition factor into three ingredients:
\begin{itemize}
\item the amplitudes \(\alpha_{kr}\) of the \(A\)-modes in sector \(F_r\),
\item the amplitudes \(\beta_{\ell r}\) of the \(B\)-modes in sector \(F_r\),
\item the temporal overlap constants
\[
\eta_r
=
\frac{1}{|\Delta|}
\int_\Delta
\chi_r^A(t)\overline{\chi_r^B(t)}\,dt.
\]
\end{itemize}
Accordingly,
\[
a_r^{(k\ell)}=\eta_r\,\alpha_{kr}\overline{\beta_{\ell r}},
\]
and the Schmidt--coherence matrix takes the form
\[
\Gamma_{rs}
=
\eta_r\overline{\eta_s}
\sum_{k,\ell}
\lambda_k\mu_\ell\,
\alpha_{kr}\overline{\alpha_{ks}}\,
\overline{\beta_{\ell r}}\beta_{\ell s}.
\]
Thus the temporal-sector approach provides a concrete constructive mechanism:
the orthogonal sectors determine the common Schmidt directions, while the
amplitudes \(\alpha_{kr}\), \(\beta_{\ell r}\) and the sector overlaps \(\eta_r\)
determine the coefficients \(a_r^{(k\ell)}\) and hence the resulting DCM state.

\section{Examples: entangled Gaussian DCM states beyond the rank-one
realization theorem}

The realization theorem of Section 6 corresponds to the special rank-one situation in which
all nonzero tensors $W_{kl}$ are collinear with a single target vector, so that the resulting
DCM state is pure. Proposition \ref{cr} shows that the common Schmidt decomposition framework
is substantially more general: one may construct tensor families $\{W_{kl}\}$ that are not collinear,
and nevertheless obtain entangled states. In this subsection we present two elementary examples in the two-qubit case.

Throughout this subsection let
\[
H_A=H_B=\mathbb C^2
\]
with fixed orthonormal bases
\[
\{e_1,e_2\}\subset H_A,\qquad \{f_1,f_2\}\subset H_B.
\]
We work in the common Schmidt basis
\[
e_1\otimes f_1,\qquad e_2\otimes f_2.
\]

\begin{example}[Bell-type mixed entangled state generated by two tensors]
Assume that only two tensors are nonzero, namely
\[
W_{11}=e_1\otimes f_1+e_2\otimes f_2,
\qquad
W_{22}=e_1\otimes f_1-e_2\otimes f_2,
\]
and that
\[
W_{k\ell}=0
\qquad\text{for all other pairs }(k,\ell).
\]
Then the Gaussian DCM covariance operator is
\[
\widehat C=\lambda_1\mu_1\,|W_{11}\rangle\langle W_{11}|
+\lambda_2\mu_2\,|W_{22}\rangle\langle W_{22}|.
\]

Expanding the two rank-one operators gives
\[
|W_{11}\rangle\langle W_{11}|
=
|e_1\otimes f_1\rangle\langle e_1\otimes f_1|
+|e_1\otimes f_1\rangle\langle e_2\otimes f_2|
+|e_2\otimes f_2\rangle\langle e_1\otimes f_1|
+|e_2\otimes f_2\rangle\langle e_2\otimes f_2|,
\]
and
\[
|W_{22}\rangle\langle W_{22}|
=
|e_1\otimes f_1\rangle\langle e_1\otimes f_1|
-|e_1\otimes f_1\rangle\langle e_2\otimes f_2|
-|e_2\otimes f_2\rangle\langle e_1\otimes f_1|
+|e_2\otimes f_2\rangle\langle e_2\otimes f_2|.
\]
Therefore
\[
\widehat  C=
(\lambda_1\mu_1+\lambda_2\mu_2)
\Bigl(
|e_1\otimes f_1\rangle\langle e_1\otimes f_1|
+
|e_2\otimes f_2\rangle\langle e_2\otimes f_2|
\Bigr)
\]
\[
\qquad\qquad
+
(\lambda_1\mu_1-\lambda_2\mu_2)
\Bigl(
|e_1\otimes f_1\rangle\langle e_2\otimes f_2|
+
|e_2\otimes f_2\rangle\langle e_1\otimes f_1|
\Bigr).
\]

Equivalently, the Schmidt--coherence matrix is
\[
\Gamma=
\begin{pmatrix}
\lambda_1\mu_1+\lambda_2\mu_2 &
\lambda_1\mu_1-\lambda_2\mu_2\\[1mm]
\lambda_1\mu_1-\lambda_2\mu_2 &
\lambda_1\mu_1+\lambda_2\mu_2
\end{pmatrix}.
\]
Hence
\[
\Gamma_{12}=\lambda_1\mu_1-\lambda_2\mu_2.
\]
By Proposition \ref{cr}, the normalized DCM state
\[
\hat \rho_C=\frac{\widehat C}{\rm{Tr}\widehat C}
\]
is entangled whenever
\[
\lambda_1\mu_1\neq \lambda_2\mu_2.
\]

Thus unequal Gaussian weights produce a genuinely mixed entangled state. If
\[
\lambda_1\mu_1=\lambda_2\mu_2,
\]
then the off-diagonal terms cancel and
\[
\hat \rho=
\frac12
\Bigl(
|e_1\otimes f_1\rangle\langle e_1\otimes f_1|
+
|e_2\otimes f_2\rangle\langle e_2\otimes f_2|
\Bigr),
\]
which is separable.
\end{example}

\begin{example}[Bell state plus a product contribution]
We now consider a family consisting of one entangled tensor and one product tensor:
\[
W_{11}=\frac{1}{\sqrt2}\bigl(e_1\otimes f_1+e_2\otimes f_2\bigr),
\qquad
W_{22}=e_1\otimes f_1,
\]
and
\[
W_{k\ell}=0
\qquad\text{for all other pairs }(k,\ell).
\]
Then
\[
\widehat  C=\lambda_1\mu_1\,|W_{11}\rangle\langle W_{11}|
+\lambda_2\mu_2\,|W_{22}\rangle\langle W_{22}|.
\]

A direct computation yields
\[
|W_{11}\rangle\langle W_{11}|
=
\frac12
\Bigl(
|e_1\otimes f_1\rangle\langle e_1\otimes f_1|
+|e_1\otimes f_1\rangle\langle e_2\otimes f_2|
+|e_2\otimes f_2\rangle\langle e_1\otimes f_1|
+|e_2\otimes f_2\rangle\langle e_2\otimes f_2|
\Bigr),
\]
whereas
\[
|W_{22}\rangle\langle W_{22}|
=
|e_1\otimes f_1\rangle\langle e_1\otimes f_1|.
\]
Therefore
\[
C=
\left(\frac12\lambda_1\mu_1+\lambda_2\mu_2\right)
|e_1\otimes f_1\rangle\langle e_1\otimes f_1|
+\frac12\lambda_1\mu_1
|e_1\otimes f_1\rangle\langle e_2\otimes f_2|
\]
\[
\qquad
+\frac12\lambda_1\mu_1
|e_2\otimes f_2\rangle\langle e_1\otimes f_1|
+\frac12\lambda_1\mu_1
|e_2\otimes f_2\rangle\langle e_2\otimes f_2|.
\]
Hence the Schmidt--coherence matrix is
\[
\Gamma=
\begin{pmatrix}
\frac12\lambda_1\mu_1+\lambda_2\mu_2 & \frac12\lambda_1\mu_1\\[1mm]
\frac12\lambda_1\mu_1 & \frac12\lambda_1\mu_1
\end{pmatrix}.
\]
In particular,
\[
\Gamma_{12}=\frac12\lambda_1\mu_1.
\]
Therefore, by Proposition~2, the normalized DCM state is entangled whenever
\[
\lambda_1\mu_1>0.
\]

Thus any nonzero Bell component already forces entanglement, even in the presence of an
additional product contribution. This example may be interpreted as a Bell state perturbed by
a separable Gaussian contribution.
\end{example}

These examples illustrate that Proposition~2 yields a broad class of entangled Gaussian DCM
states that are qualitatively different from the rank-one realization theorem. In the realization
theorem, all nonzero tensors \(W_{k\ell}\) are proportional to a single vector and the resulting
state is pure. In contrast, the present examples involve several linearly independent tensors
sharing a common Schmidt basis, so that the DCM state is mixed. Entanglement is then encoded
in the off-diagonal entries of the Schmidt--coherence matrix \(\Gamma\).

\begin{example}[Temporal-sector realization of a Bell-type mixed state]

We return to the Bell-type mixed family considered above, namely
\[
W_{11}=e_1\otimes f_1+e_2\otimes f_2,
\qquad
W_{22}=e_1\otimes f_1-e_2\otimes f_2,
\]
and \(W_{k\ell}=0\) for all other pairs \((k,\ell)\),
and consider examples of  the corresponding temporal sector realizations. 

Let the temporal Hilbert space admit an orthogonal decomposition
\[
L^2(\Delta)=F_1\oplus F_2\oplus F_\perp,
\]
where \(F_1\) and \(F_2\) are closed mutually orthogonal subspaces. Choose
functions
\[
\chi_1^A,\chi_1^B\in F_1,
\qquad
\chi_2^A,\chi_2^B\in F_2,
\]
and define the sector overlap constants
\[
\eta_r
=
\frac1{|\Delta|}
\int_\Delta \chi_r^A(t)\overline{\chi_r^B(t)}\,dt,
\qquad r=1,2.
\]

Assume for simplicity that
\[
\eta_1=\eta_2=1.
\]
For instance, this holds if one chooses normalized profiles with
\(\chi_r^A=\chi_r^B\), but this equality is not required in the general framework.

Define the temporal modes by
\[
u_1(t)=\chi_1^A(t)e_1+\chi_2^A(t)e_2,
\qquad
v_1(t)=\chi_1^B(t)f_1+\chi_2^B(t)f_2,
\]
\[
u_2(t)=\chi_1^A(t)e_1+\chi_2^A(t)e_2,
\qquad
v_2(t)=\chi_1^B(t)f_1-\chi_2^B(t)f_2.
\]
Equivalently, the nonzero amplitudes are
\[
\alpha_{11}=\alpha_{12}=\alpha_{21}=\alpha_{22}=1,
\]
\[
\beta_{11}=1,\qquad \beta_{12}=1,\qquad
\beta_{21}=1,\qquad \beta_{22}=-1.
\]

By Corollary~\ref{PairedSectorProfiles},
\[
W_{11}
=
\eta_1\alpha_{11}\overline{\beta_{11}}\,e_1\otimes f_1
+
\eta_2\alpha_{12}\overline{\beta_{12}}\,e_2\otimes f_2
=
e_1\otimes f_1+e_2\otimes f_2,
\]
and
\[
W_{22}
=
\eta_1\alpha_{21}\overline{\beta_{21}}\,e_1\otimes f_1
+
\eta_2\alpha_{22}\overline{\beta_{22}}\,e_2\otimes f_2
=
e_1\otimes f_1-e_2\otimes f_2.
\]
Thus the Bell-type mixed state is realized by using the same two temporal
sectors for the two Schmidt directions and flipping the sign of the second
sector in the second \(B\)-mode.
\end{example}

\begin{example}[Temporal-sector realization of a Bell-plus-product family]
We now realize the family
\[
W_{11}=\frac1{\sqrt2}(e_1\otimes f_1+e_2\otimes f_2),
\qquad
W_{22}=e_1\otimes f_1,
\]
and \(W_{k\ell}=0\) for all other pairs \((k,\ell)\).

Let
\[
L^2(\Delta)=F_1\oplus F_2\oplus F_\perp,
\]
with \(F_1\perp F_2\), and choose functions
\[
\chi_1^A,\chi_1^B\in F_1,
\qquad
\chi_2^A,\chi_2^B\in F_2.
\]
Define
\[
\eta_r
=
\frac1{|\Delta|}
\int_\Delta \chi_r^A(t)\overline{\chi_r^B(t)}\,dt,
\qquad r=1,2,
\]
and assume again for simplicity that
\[
\eta_1=\eta_2=1.
\]

Define the temporal modes by
\[
u_1(t)=\frac1{\sqrt2}\chi_1^A(t)e_1+\frac1{\sqrt2}\chi_2^A(t)e_2,
\qquad
v_1(t)=\chi_1^B(t)f_1+\chi_2^B(t)f_2,
\]
\[
u_2(t)=\chi_1^A(t)e_1,
\qquad
v_2(t)=\chi_1^B(t)f_1.
\]
Equivalently, the nonzero amplitudes are
\[
\alpha_{11}=\frac1{\sqrt2},\qquad \alpha_{12}=\frac1{\sqrt2},
\qquad
\beta_{11}=1,\qquad \beta_{12}=1,
\]
\[
\alpha_{21}=1,\qquad \alpha_{22}=0,
\qquad
\beta_{21}=1,\qquad \beta_{22}=0.
\]

Again by Corollary~\ref{PairedSectorProfiles},
\[
W_{11}
=
\eta_1\alpha_{11}\overline{\beta_{11}}\,e_1\otimes f_1
+
\eta_2\alpha_{12}\overline{\beta_{12}}\,e_2\otimes f_2
=
\frac1{\sqrt2}(e_1\otimes f_1+e_2\otimes f_2),
\]
while
\[
W_{22}
=
\eta_1\alpha_{21}\overline{\beta_{21}}\,e_1\otimes f_1
+
\eta_2\alpha_{22}\overline{\beta_{22}}\,e_2\otimes f_2
=
e_1\otimes f_1.
\]
Hence the Bell-plus-product family is obtained by suppressing the second
temporal sector in the second pair of modes.
\end{example}

These examples illustrate concretely the mechanism behind
Proposition~\ref{Tsector}, Corollary~\ref{PairedSectorProfiles}, and
Proposition~2. The orthogonal temporal sectors \(F_r\) determine the Schmidt
directions \(e_r\otimes f_r\), while the amplitudes \(\alpha_{kr}\),
\(\beta_{\ell r}\), together with the sector overlap constants \(\eta_r\),
determine the coefficients
\[
a_r^{(k\ell)}=\eta_r\,\alpha_{kr}\overline{\beta_{\ell r}}.
\]
Consequently,
\[
\Gamma_{rs}
=
\eta_r\overline{\eta_s}
\sum_{k,\ell}\lambda_k\mu_\ell\,
\alpha_{kr}\overline{\alpha_{ks}}\,
\overline{\beta_{\ell r}}\beta_{\ell s}.
\]
Thus the temporal-sector decomposition provides a concrete constructive route
from Gaussian processes to a broad family of mixed entangled DCM states, well
beyond the rank-one pure-state realization theorem.

For simplicity, we realize these examples within the special framework of Corollary, where all mode components in a given sector are proportional to fixed profiles $\chi_r^A(t)$ and $\chi_r^B(t).$ This is not essential for the existence of the examples, but 
it yields a particularly transparent construction. Below we present schematically the general framework.

\paragraph{Different temporal profiles in the same sector.}

The paired-profile assumption of Corollary~\ref{PairedSectorProfiles} is only a
convenient sufficient condition for obtaining explicit coefficient formulas.
It is not necessary for the realization of entangled families \(\{W_{k\ell}\}\).

To illustrate this, suppose that two temporal sectors \(F_1,F_2\subset L^2(\Delta)\)
are fixed. Choose functions
\[
\phi_{11},\phi_{21}\in F_1,
\qquad
\psi_{11},\psi_{21}\in F_1,
\]
and
\[
\phi_{12},\phi_{22}\in F_2,
\qquad
\psi_{12},\psi_{22}\in F_2.
\]
Define the temporal modes by
\[
u_1(t)=\phi_{11}(t)e_1+\phi_{12}(t)e_2,
\qquad
u_2(t)=\phi_{21}(t)e_1+\phi_{22}(t)e_2,
\]
and
\[
v_1(t)=\psi_{11}(t)f_1+\psi_{12}(t)f_2,
\qquad
v_2(t)=\psi_{21}(t)f_1+\psi_{22}(t)f_2.
\]

Then Proposition~\ref{Tsector} yields
\[
W_{k\ell}
=
a_1^{(k\ell)}\,e_1\otimes f_1
+
a_2^{(k\ell)}\,e_2\otimes f_2,
\]
where the coefficients are given by the temporal overlaps
\[
a_1^{(11)}
=
\frac{1}{|\Delta|}
\int_\Delta
\phi_{11}(t)\overline{\psi_{11}(t)}\,dt,
\qquad
a_2^{(11)}
=
\frac{1}{|\Delta|}
\int_\Delta
\phi_{12}(t)\overline{\psi_{12}(t)}\,dt,
\]
\[
a_1^{(22)}
=
\frac{1}{|\Delta|}
\int_\Delta
\phi_{21}(t)\overline{\psi_{21}(t)}\,dt,
\qquad
a_2^{(22)}
=
\frac{1}{|\Delta|}
\int_\Delta
\phi_{22}(t)\overline{\psi_{22}(t)}\,dt.
\]

Hence, in order to realize the Bell-type family
\[
W_{11}=e_1\otimes f_1+e_2\otimes f_2,
\qquad
W_{22}=e_1\otimes f_1-e_2\otimes f_2,
\]
it is enough to choose these four overlaps so that
\[
a_1^{(11)}=1,\qquad
a_2^{(11)}=1,\qquad
a_1^{(22)}=1,\qquad
a_2^{(22)}=-1.
\]
There is no requirement that the temporal profiles coincide across the different
modes. In particular, one does not need
\[
\phi_{11}=\phi_{21},
\qquad
\phi_{12}=\phi_{22},
\]
nor
\[
\psi_{11}=\psi_{21},
\qquad
\psi_{12}=\psi_{22}.
\]

The price of allowing such freedom is that one loses the simple factorization
from Corollary~\ref{PairedSectorProfiles}. In the paired-profile setting, the
coefficients factor as
\[
a_r^{(k\ell)}=\eta_r\,\alpha_{kr}\overline{\beta_{\ell r}},
\]
so that each sector contributes a rank-one coefficient array in the indices
\((k,\ell)\). In the present more general setting, each coefficient is instead
an independent temporal overlap,
\[
a_r^{(k\ell)}
=
\frac{1}{|\Delta|}
\int_\Delta
\phi_{kr}(t)\overline{\psi_{\ell r}(t)}\,dt
=
\frac{1}{|\Delta|}\,
\langle \psi_{\ell r},\phi_{kr}\rangle_{L^2(\Delta)}.
\]
Thus the paired-profile ansatz provides a particularly transparent
rank-one factorization in each sector, whereas the general temporal-sector
framework allows a much larger class of coefficient arrays.

This is a good place to emphasize once again that, in DCM, large classes of
classical states (stochastic processes) may give rise to the same quantum state. 

\section{Conclusion and Outlook}

The Double Covariance Model (DCM) was introduced as a classical probabilistic
framework in which quantum states of composite systems are generated from
classical stochastic processes by means of a \emph{double covariance}
construction, that is, as covariances of random temporal covariances. In its
full generality, this construction is genuinely fourth order: the DCM state is
determined by the covariance of the random tensor
\[
C_\Delta(\omega)=\frac1{|\Delta|}\int_\Delta X_t(\omega)\otimes Y_t(\omega)\,dt,
\]
or, equivalently, by the covariance of the corresponding random operator
\[
\widehat C_\Delta(\omega)=\frac1{|\Delta|}\int_\Delta |X_t(\omega)\rangle\langle
Y_t(\omega)|\,dt.
\]
Thus the DCM naturally lives in a fourth-order statistical framework.

The main result of the present paper is that, for jointly Gaussian processes,
this fourth-order structure admits a complete reduction to second-order
covariance data. By the Isserlis--Wick theorem, every fourth-order moment
entering the DCM construction is uniquely expressed in terms of pair
covariances. Consequently, the double-covariance operator, and therefore the
associated quantum density operator, is completely determined by the second-order
covariance kernel of the joint Gaussian process. In this sense, the Gaussian
reduction transforms the DCM from a fourth-order statistical model into a
second-order one without changing the resulting quantum state.

This reduction is not merely a technical simplification. It leads to a concrete
and flexible Gaussian realization theory for quantum states generated by the
DCM. In particular, by working with the Karhunen--Lo\`eve expansions of the
subquantum Gaussian processes, one obtains an explicit description of the
time-averaged tensor amplitudes
\[
W_{k\ell}
=
\frac1{|\Delta|}\int_\Delta u_k(t)\otimes v_\ell(t)\,dt,
\]
where \(u_k\) and \(v_\ell\) are the Karhunen--Lo\`eve modes of the two
processes. These tensors form the basic deterministic objects of the Gaussian
DCM construction. In the independent Gaussian case, the DCM covariance operator
takes the form
\[
\widehat C
=
\sum_{k,\ell}\lambda_k\mu_\ell\,|W_{k\ell}\rangle\langle W_{k\ell}|,
\]
so that the resulting quantum state is completely determined by the geometry of
the tensor family \(\{W_{k\ell}\}\).

This point of view reveals a feature of the DCM that is both mathematically
striking and conceptually counterintuitive from the perspective of standard
quantum-information intuition: quantum entangled states can be generated by
\emph{independent} jointly Gaussian processes. Thus, within the DCM, the origin
of entanglement need not lie in ordinary second-order statistical dependence
between the subquantum processes \(X\) and \(Y\). In the Gaussian framework, the
decisive ingredient is instead the temporal organization of the Karhunen--Lo\`eve
modes and, more specifically, the structure of the integrated tensors
\(W_{k\ell}\).

For pure-state realizations, this temporal organization is naturally encoded by
the decomposition of the temporal Hilbert space into mutually orthogonal
sectors. The temporal-sector construction developed in this paper shows that the
Schmidt structure of the target quantum state may be realized through an
appropriate allocation of the Gaussian modes to orthogonal temporal sectors. In
this way, the temporal-sector decomposition emerges as a basic constructive
mechanism for the generation of entanglement in the Gaussian DCM.

The same Gaussian framework also gives a new interpretation of concurrence. For
the pure-state realizations considered in this paper, the Schmidt coefficients
are directly related to the expected distribution of the subquantum signal
energy among the temporal sectors. Hence concurrence can be expressed in terms
of the relative sector energies. We stress that this ``energy'' is not the
physical energy observable of the quantum system itself. Rather, it is the
energy of the classical random signals at the subquantum time scale, that is,
an internal quantity of the DCM representation. Within this representation,
entanglement becomes linked to a specific pattern of energy redistribution among
orthogonal temporal sectors.

Beyond the rank-one realization theorem for pure states, the Gaussian reduction
also suggests a broader general theory of Gaussian DCM states. In the
independent Gaussian case, the entire DCM state is encoded by the tensor family
\(\{W_{k\ell}\}\), and this shifts the emphasis from covariance kernels
themselves to the geometry of time-averaged tensor modes. This perspective led
us to the study of families \(\{W_{k\ell}\}\) admitting a common Schmidt
decomposition and to the corresponding Schmidt--coherence matrix. In that
framework, the temporal-sector decomposition provides a natural sufficient
condition for the existence of a common Schmidt decomposition and hence a
constructive route to broad classes of mixed entangled Gaussian DCM states.

At the same time, the present results should be interpreted with appropriate
caution. Neither the DCM itself nor its Gaussian reduction can at present be
regarded as a complete classical probabilistic foundation of quantum mechanics.
What has been established so far is more modest, though still substantial. On
the one hand, the DCM provides a mathematically explicit classical mechanism for
the generation of quantum entangled states, including states produced by
independent Gaussian processes. On the other hand, earlier work has shown that
the same framework can be extended to quantum Markovian dynamics. These results
capture important structural aspects of quantum theory, but they are far from
constituting a full reconstruction of quantum mechanics.

Nevertheless, the DCM and its Gaussian reduction open a promising direction for
research on the interplay between classical and quantum probability. Their main
novelty lies in the use of a genuinely double-covariance construction based on
two distinct time scales and on the interplay between temporal and statistical
covariances. From this perspective, quantum states are not represented by
ordinary second-order classical statistics, but by statistical covariances of
random temporal covariances. The Gaussian reduction shows that this fourth-order
framework can nevertheless be controlled by second-order covariance data, while
still retaining a rich entanglement structure. We expect that further
development of this viewpoint may lead to a broader geometric and probabilistic
theory of Gaussian DCM states and, more generally, to new insights into the
classical probabilistic structures underlying quantum phenomena.

\section*{Acknowledgments} 

This study was stimulated by a remark made by Lev Murokh during the discussion following my talk at the conference QIP26 (V\"axj\"o, Sweden, June 2026). He appreciated the novelty of the DCM as a fourth-order statistical framework, but pointed out that in the Gaussian case it should admit a reduction to second-order statistics. I am grateful to him for this insightful observation, which motivated me to write the present paper. I would also like to thank the other participants of QIP26 for valuable discussions and helpful comments.

\section*{Appendix: Isserlis--Wick Reduction for Circular Gaussian Processes}

\paragraph{On the Historical and Structural Equivalence of Isserlis's and Wick's Theorems.}

In the physics literature, the algebraic reduction of higher-order statistical moments to combinatoric sums of products of second-order covariances is almost universally referred to as \textit{Wick's theorem}. In its original 1950 formulation \cite{wick1950}, Wick's theorem is strictly an operator-theoretic result designed for quantum field theory, dictating how products of non-commuting creation and annihilation operators on a Fock space decompose into sums of pairwise contractions. 

However, later mathematical physicists recognized a profound structural parallel: the combinatorics governing quantum vacuum expectation values of field operators are formally identical to the combinatorics of calculating higher-order moments of classical Gaussian random variables (see, e.g., \cite{simon1974,janson1997}). Because physicists utilize quantum field theory daily, the name \textit{Wick's theorem} evolved into a ubiquitous blanket term within the physics community for any algebraic procedure that partitions a fourth-order (or higher) product into pairwise contractions---regardless of whether the underlying system is classical or quantum. Consequently, even in purely classical domains such as statistical mechanics or stochastic signal processing, physicists heavily favor the invocation of ``Wick's theorem.''

Despite this widespread convention, DCM is established within a strictly classical probabilistic framework governed by Kolmogorov probability spaces and commuting, complex-valued stochastic processes. To maintain mathematical precision and emphasize the purely classical foundations of our model, we invoke the proper historical designation from multivariate statistics: \textit{Isserlis's theorem} (originally proved by Leon Isserlis in 1918 
\cite{isserlis1918,peccati2011}). 

Crucially, when Isserlis's theorem is restricted to the \textit{centered circular complex} Gaussian processes used in this work, phase-invariance symmetry forces all pure non-conjugated and pure conjugated pseudocovariances ($\mathbb{E}[X_t X_s] = 0$) to vanish identically. This classical probabilistic condition isolates a set of surviving cross-contractions that is combinatorially identical to the behavior of field operators in quantum mechanics, where non-matching operator pairings vanish ($\langle 0 | aa | 0 \rangle = 0$). Utilizing Isserlis's theorem under circularity thus preserves the purely classical foundation of our model while cleanly recovering the algebraic reduction rules required to generate entangled quantum states, showing that this specific contraction structure is an intrinsic feature of circular Gaussianity rather than an exclusively quantum phenomenon.

\begin{theorem}
Let $X_t$ and $Y_t$ be centered circular Gaussian Hilbert-space-valued stochastic processes. Then the fourth-order tensor kernel $\hat{K}(t,s) = \mathbb{E}[|X_t \otimes Y_t\rangle\langle X_s \otimes Y_s|]$ admits a complete Wick reduction to the second-order covariance kernels given by:
\begin{equation}\label{eq:wick_reduction}
\hat{K}(t,s) = \hat{R}_{XX}(t,s) \otimes \hat{R}_{YY}(t,s) + \hat{E}(t,s),
\end{equation}
where $\hat{E}(t,s)$ is the exchange operator uniquely defined by its matrix elements:
\begin{equation}
\langle u \otimes v, \hat{E}(t,s)(u' \otimes v')\rangle = \langle u, \hat{R}_{XY}(t,s)v'\rangle \langle v, \hat{R}_{YX}(t,s)u'\rangle.
\end{equation}
\end{theorem}

\begin{proof}
Let $u, u' \in H_A$ and $v, v' \in H_B$ be arbitrary vectors. The matrix elements of the fourth-order operator $\hat{K}(t,s)$ are given by the statistical expectation:
\begin{equation}\label{eq:matrix_element}
\langle u \otimes v, \hat{K}(t,s)(u' \otimes v')\rangle = \mathbb{E}[\langle u, X_t \rangle \langle v, Y_t \rangle \overline{\langle u', X_s \rangle \langle v', Y_s \rangle}].
\end{equation}
We map this expression onto the classical Isserlis--Wick theorem for complex random variables by setting:
\begin{align*}
A &= \langle u, X_t \rangle, & B &= \langle v, Y_t \rangle, \\
\bar{C} &= \overline{\langle u', X_s \rangle}, & \bar{D} &= \overline{\langle v', Y_s \rangle}.
\end{align*}
The full combinatorial pairing expansion under the Wick theorem yields exactly three distinct structural combinations:
\begin{equation}\label{eq:full_wick}
\mathbb{E}[A B \bar{C} \bar{D}] = \mathbb{E}[A \bar{C}] \mathbb{E}[B \bar{D}] + \mathbb{E}[A \bar{D}] \mathbb{E}[B \bar{C}] + \mathbb{E}[A B] \mathbb{E}[\bar{C} \bar{D}].
\end{equation}
We analyze each pairing separately using the definition of the second-order covariance kernels:
\begin{enumerate}
    \item \textbf{Pure Subsystem Pairings:} The first term pairs the non-conjugated component of each subsystem with its own macroscopic conjugated counterpart:
    \begin{equation}
    \mathbb{E}[A \bar{C}] = \mathbb{E}[\langle u, X_t \rangle \overline{\langle u', X_s \rangle}] = \langle u, \hat{R}_{XX}(t,s) u' \rangle,
    \end{equation}
    \begin{equation}
    \mathbb{E}[B \bar{D}] = \mathbb{E}[\langle v, Y_t \rangle \overline{\langle v', Y_s \rangle}] = \langle v, \hat{R}_{YY}(t,s) v' \rangle.
    \end{equation}
    Taking their product directly reconstructs the standard tensor product operator representation:
    \begin{equation}
		\label{eq:term1}
    \mathbb{E}[A \bar{C}] \mathbb{E}[B \bar{D}] = \langle u \otimes v, \left( \hat{R}_{XX}(t,s) \otimes \hat{R}_{YY}(t,s) \right) (u' \otimes v') \rangle.
    \end{equation}

    \item \textbf{Cross-Subsystem Pairings:} The second term pairs the temporal state of system $A$ with system $B$, describing the spatial correlation structure:
    \begin{equation}
    \mathbb{E}[A \bar{D}] = \mathbb{E}[\langle u, X_t \rangle \overline{\langle v', Y_s \rangle}] = \langle u, \hat{R}_{XY}(t,s) v' \rangle,
    \end{equation}
    \begin{equation}
    \mathbb{E}[B \bar{C}] = \mathbb{E}[\langle v, Y_t \rangle \overline{\langle u', X_s \rangle}] = \langle v, \hat{R}_{YX}(t,s) u' \rangle.
    \end{equation}
    By utilizing the definition of the cross-system exchange operator $\hat{E}(t,s)$, this structural block evaluates to:
    \begin{equation}\label{eq:term2}
    \mathbb{E}[A \bar{D}] \mathbb{E}[B \bar{C}] = \langle u \otimes v, \hat{E}(t,s) (u' \otimes v') \rangle.
    \end{equation}

    \item \textbf{Pseudocovariance Pairings:} The third possible pairing collects terms of matching conjugation types:
    \begin{equation*}
    \mathbb{E}[A B] = \mathbb{E}[\langle u, X_t \rangle \langle v, Y_t \rangle], \quad \text{and} \quad \mathbb{E}[\bar{C} \bar{D}] = \mathbb{E}[\overline{\langle u', X_s \rangle} \, \overline{\langle v', Y_s \rangle}].
    \end{equation*}
    Because the underlying processes are explicitly defined as circular Gaussian under Definition 1, all pure non-conjugated and pure conjugated pairings collapse identically to zero due to phase-invariance symmetry:
    \begin{equation*}
    \mathbb{E}[A B] = 0 \quad \text{and} \quad \mathbb{E}[\bar{C} \bar{D}] = 0.
    \end{equation*}
    Consequently, this entire combinatorial channel drops out of the expectation profile:
    \begin{equation}\label{eq:term3}
    \mathbb{E}[A B] \mathbb{E}[\bar{C} \bar{D}] = 0 \cdot 0 = 0.
    \end{equation}
\end{enumerate}
Substituting Equations \eqref{eq:term1}, \eqref{eq:term2}, and \eqref{eq:term3} back into the Wick structural expansion \eqref{eq:full_wick} simplifies the matrix element expression to:
\begin{equation*}
\langle u \otimes v, \hat{K}(t,s)(u' \otimes v')\rangle = \langle u \otimes v, \left(\hat{R}_{XX}(t,s) \otimes \hat{R}_{YY}(t,s) + \hat{E}(t,s)\right)(u' \otimes v')\rangle.
\end{equation*}
Since this identity holds tightly for all test vectors on the arbitrary dense domains of $H_A \otimes H_B$, it implies the operator identity $\hat{K}(t,s) = \hat{R}_{XX}(t,s) \otimes \hat{R}_{YY}(t,s) + \hat{E}(t,s)$.
\end{proof}

\end{document}